\documentclass[10pt]{article}
\usepackage[nofonts]{dgleich-common}

\usepackage[T1]{fontenc}
\RequirePackage[scaled=0.8]{helvet}%
\RequirePackage[scaled=0.75]{beramono}%
\def\lstbasicfont{\fontfamily{fvm}\selectfont}
\makeatletter  
\setboolean{@dgleich@sanssmallcaps}{true}
\makeatother

\usepackage{dgleich-math}
\usepackage{dgleich-tensor}

\graphicspath{{./images/}}
\usenatbib
\usehyperref

\newcommand\nth{\textsuperscript{th}\xspace}
\newcolumntype{Z}{>{\RaggedRight\arraybackslash}X}
\DeclareMathOperator{\EX}{\mathbb{E}}% expected value
\DeclareMathOperator{\PR}{\mathbb{P}}% probability
\newenvironment{claim}[1]{\par\noindent{Claim:}\space#1}{}



\begin{document}

\title{Triangle Preferential Attachment Has Power-law Degrees and Eigenvalues; Eigenvalues Are More Stable to Network Sampling}
\author{Nicole Eikmeier, Purdue University \\ David F. Gleich, Purdue University}
%\title{Degree \emph{\&} Eigenvalue Power-laws in Triangle Preferential Attachment}
%\title{Triangle Preferential Attachment: Power-law Degrees \& Eigenvalues}
%\title{Higher Order Preferential Attachment Graphs}
%\titlenote{Produces the permission block, and
%	copyright information}
%\subtitlenote{The full version of the author's guide is available as
%	\texttt{acmart.pdf} document}

\maketitle

\begin{abstract}
Preferential attachment models are a common class of graph models which have been used to explain why power-law distributions appear in the degree sequences of real network data. One of the things they lack, however, is higher-order network clustering, including non-trivial clustering coefficients. In this paper we present a specific Triangle Generalized Preferential Attachment Model (TGPA) that, by construction, has nontrivial clustering. We further prove that this model has a power-law in both the degree distribution and eigenvalue spectra. We use this model to investigate a recent finding that power-laws are more reliably observed in the eigenvalue spectra of real-world networks than in their degree distribution. One conjectured explanation for this is that the spectra of the graph is more robust to various sampling strategies that would have been employed to collect the real-world data compared with the degree distribution. 
Consequently, we generate random TGPA models that provably have a power-law in both, and sample subgraphs via forest fire, depth-first, and random edge models. We find that the samples show a power-law in the spectra even when only 30\% of the network is seen. Whereas there is a large chance that the degrees will not show a power-law. Our TGPA model shows this behavior much more clearly than a standard preferential attachment model. This provides one possible explanation for why power-laws may be seen frequently in the spectra of real world data.
\end{abstract}

%
% The code below should be generated by the tool at
% http://dl.acm.org/ccs.cfm
% Please copy and paste the code instead of the example below.
%
%\begin{CCSXML}
%	<ccs2012>
%	<concept>
%	<concept_id>10003752.10010061.10010069</concept_id>
%	<concept_desc>Theory of computation~Random network models</concept_desc>
%	<concept_significance>500</concept_significance>
%	</concept>
%	<concept>
%	<concept_id>10010147.10010341.10010342.10010344</concept_id>
%	<concept_desc>Computing methodologies~Model verification and validation</concept_desc>
%	<concept_significance>300</concept_significance>
%	</concept>
%	</ccs2012>
%\end{CCSXML}
%
%\ccsdesc[500]{Theory of computation~Random network models}
%\ccsdesc[300]{Computing methodologies~Model verification and validation}
%
%
%\keywords{Preferential Attachment}

\marginnote[18\baselineskip]{Research supported in part by
	NSF CAREER CCF-1149756, IIS-1422918, IIS-1546488, NSF Center for Science of Information STC, CCF-0939370, DOE  DE-SC0014543, DARPA SIMPLEX, and the Sloan Foundation.
}

\section{Introduction}
\label{sec:Intro}

The idea of preferential attachment (PA) has a lengthy history in explaining ``rich-get-richer'' models~\cite{Yule-1925-power-law,Price-1976-preferential-attachment}. In the context of networks, a preferential attachment model suggests that when agents join a network, they form links to existing nodes with large degrees. These models offer a simple local rule that helps explain the presence of highly-skewed or power-law degree distributions in real-world networks~\cite{barabasi1999emergence}. While a simple and compelling mathematical model, there are weaknesses in the relationship between PA models and real-world data. One of the most striking is the lack of clustering in PA network models. Consequently, there has been a line of work on generalized PA models that include ways to address the lack of clustering. First,~\citet{holme2002growing} proposed a triangle PA model, where agents arrive and link to a node based on its degree and also link to a neighbor of that node to form a triangle. Later,~\citet{ostroumova2013generalized} generalized a family of PA models and showed that they had power-law degree distributions and high-clustering.

Our work follows in this vein, although we adapt a slightly different notion of a triangle PA model that builds on a recent proposal to show how preferential attachment could give a power-law with any exponent~\cite{avin2017improved}. The specific \emph{Triangle Generalized Preferential Attachment} model we use has two slightly different forms as explained in Section~\ref{sec:TGPA}. The two forms are used to greatly simplify the analysis of the resulting properties. We do not believe there to be qualitative differences between them. Formally, we show that these models have a power-law in the degree distribution (Theorem~\ref{TGPATheorem1}, Corollary~\ref{corr:tpga-power-law}) as well as a power-law in the eigenvalues of the adjacency matrix (Theorem~\ref{TGPATheorem2}).

We also find empirically that our TGPA model has higher-order clustering in terms of higher-order clique closures~\cite{yin2017higher} that is characteristic of real-world data (Section~\ref{sec:clust}).

Our interest in the TGPA model stems from our recent finding on the reliable presence of power-laws in the eigenvalue spectrum of the adjacency matrix~\cite{Eikmeier-2017-power-laws}. Specifically,~\citet{Eikmeier-2017-power-laws} found that real-world networks of a variety of types were more likely to have a statistically significant power-law in the eigenvalues of the adjacency matrix than in the degree distribution. This observation presents a simple question, might this behavior be expected in light of how real-world network data are collected? To be specific, real-world network data reflect two types of sampling artifacts. They are often built from a process run on a larger dataset. Consider how web and social networks are often \emph{crawled} by programs that use breadth-first or related crawling strategies. Second, the crawled data itself represents a sample of some underlying (and unknown) latent network~\cite{schoenebeck2013potential}. Again, note that the social links on networks such as Facebook and Twitter only represent a sample of some unobserved \emph{true} social network. Because of the way that individuals join these networks, forest-fire models are often used to simulate this type of artifact.

Consequently, we study how often \emph{samples} of TGPA models have statistically significant power laws in their degrees and eigenvalues (Section~\ref{sec:per}). These results (Figure~\ref{fig:per}) offer compelling evidence that the eigenvalues of the adjacency matrix robustly indicate the presence of a power-law, with more reliability than the degrees. It should be noted that the presence (or lack thereof) of power-laws in real world data has been often debated~\cite{MeuselVignaLehmbergEtAl2015,GjokaKurantButtsEtAl2010,Clauset-2018-pl-are-rare}.  For that reason, we study models where they are unambiguously present. Although other PA models have the needed property of power-laws in both spectra and degrees, we find that the differences in behavior between the sampled eigenvalues and degrees are less clear than in TGPA.

In summary, the primary contributions of this manuscript are:
\begin{compactenum}
	\item We extend the results presented on the Generalized Preferential Attachment Model (in~\citet{avin2017improved}), to show the eigenvalues follow a power-law distribution. (Section~\ref{sec:GPAResults})
	\item We present the Triangle Generalized Preferential Attachment Model (TGPA): A model which imposes higher order structure directly into the network. (Section~\ref{sec:TGPA})
	\item We conduct extensive analysis of TGPA to show that the degrees follow a power-law distribution with an exponent which can range between $(1,\infty)$ (Section~\ref{sec:DegreeTGPA}), and that the eigenvalues follow a power-law distribution. (Section~\ref{sec:SpectraTGPA})
	%\item We show that TGPA produces graphs with significant clustering. (Section~\ref{sec:clust})
	\item We use TGPA to support a conjecture on why power-laws are observed more often in spectra of networks, and study the results of perturbing the TGPA model. (Section~\ref{sec:per})
\end{compactenum}
%In the interest of leaving space to keep complete proofs, we omit a formal conclusion section and intersperse concluding comments throughout.

\section{Preliminaries and Related Work}
\label{sec:related}

Denote a graph $G$ by its set of vertices $V$ and edges $E$. A graph with $n$ vertices can be represented as an $n \times n$ adjacency matrix $\mA$, where $\mA_{ij} = 1$ if edge $(i,j)$ is in the graph, and $\mA_{ij} = 0$ otherwise. The degree of vertex $i$ is the number of vertices $j$ such that $\mA_{ij} = 1$.

We will be concerned with graph models that evolve over time. There are a huge diversity of graph generation schemes, many of which have been analyzed in theory and in practice. For example latent space graphs~\cite{lattanzi2009affiliation} and Geometric Preferential Attachment. Start with some set of vertices and edges, $G_0 = (V_0,E_0)$. At each time step $t = 1,2,\ldots$ perform some action on $G_{t-1}$ (such as adding new vertices or edges) to obtain $G_t = (V_t, E_t)$. Continue until the graph is sufficiently large. Denote the degree of vertex $v$ at time $t$ to be $d_t(v)$. Let $e_t$ denote the number of edges at time $t$, and let $m_{k,t}$ be the number of nodes at time $t$ with degree $k$.

\subsection{Preferential Attachment}
\label{sec:PA}
Preferential attachment (PA) describes a mechanism of graph evolution in which nodes with higher degree tend to continue gaining neighbors. When a new node $u$ is added to the graph at time $t$, choose another existing vertex $v$ with probability proportional to its degree. Formally, choose vertex $v$ with probability 
\begin{equation}
\label{eqn:gamma}
\gamma_{t}(v) =  \frac{d_{t-1}(v)}{\sum_{w \in V_{t-1}} d_{t-1}(w)}.
\end{equation}
Then add an edge connecting $u$ to $v$. PA is meant to model the \emph{power-law} behavior that is often seen in real-world networks~\cite{FaloutsosFaloutsosFaloutsos1999,Huberman2001,MedinaMattaByers2000}, that is a few vertices tend to have very large degree while most vertices have fairly low degree. A set of values $ x_1, x_2, \ldots x_k $ satisfies a power-law if it is drawn from a probability distribution $p(x) \propto x^{-\beta}$ for some $\beta$.

The PA graph model is found in a few different forms. In the model by~\citet{barabasi1999emergence}, often called the BA model, at every new time step, a new vertex is formed with $m$ edges. Each of the edges is then connected to an existing node chosen using PA, i.e. based on their degrees.

In a slight variation found in~\citet{chung2006complex,cooper2001general}, at each time step $t$, a new node is added with probability $p$. Along with the new node is an edge between the new node and an existing node picked via PA. With probability $1-p$ a new edge is added between two existing nodes, both chosen via PA.
These two models generate slightly different distributions, but fundamentally give very similar graphs. We present our model TGPA in two forms matching these differences (Section~\ref{sec:TGPA}).

In the next sections we discuss a few variations of the PA model. (sections~\ref{sec:GPA},~\ref{sec:Holme}). There exist other variations of PA~\cite{toivonen2006model,zadorozhnyi2015growing,ostroumova2013generalized,saramaki2004scale} which we will not detail here.

\subsection{Generalized Preferential Attachment}
\label{sec:GPA}
The Generalized Preferential Attachment Model (GPA) was defined by~\citet{avin2017improved}. In this model, in addition to adding new vertices and edges, there is also an option in each time step of adding a new \emph{component}. Furthermore, the parameters may change over time, if desired. Start with an arbitrary initial non-empty graph $G_0$. For time $t \geq 1$, the graph $G_t$ is constructed by performing either a \textit{node event} with probability $p_t \in [0,1]$, an \textit{edge event} with probability $r_t \in [0,1-p_t]$, or a \textit{component event} with probability $q_t = 1 - p_t - r_t$. In a node event, a new vertex $v$ is added to the graph, along with an edge $(u,v)$ where $u$ is chosen from $G_{t-1}$ with probability 
$\gamma_t(u)$. In an edge event, a new edge $(u,w)$ is added, with $u$ and $w$ both nodes in $G_{t-1}$, and they are chosen with probability $\gamma_t(u) \cdot \gamma_w(u)$. And in a component event, two new nodes $v_1,v_2$ are added along with edge $(v_1,v_2)$. Exactly one edge is added at each time step, so the number of edges in $G_t$ is  equal to $e_0 + t$.% and the sum of degrees is $2e_t = 2(e_0 + t)$. 

The key difference of this model defined by~\citet{avin2017improved} over the PA model discussed in Section \ref{sec:PA} is the ability to add new components to the graph. In~\citet{avin2017improved}, it is proved that the degree distribution follows a power-law. In this manuscript we further prove that the eigenvalues follow a power-law distribution. (See Section \ref{sec:GPAResults}).

We will also work with a slight variation of the GPA model, along the lines of the alternate version of the PA model defined in~\citet{flaxman2005high,barabasi1999emergence} and discussed in Section~\ref{sec:PA}. Start with an empty graph. At time $t = 1,2,\ldots$ do one of the following: 
With probability $p$ add a new vertex $v_t$ and an edge from $v_t$ to some other vertex in $u$ where $u$ is chosen with probability
\begin{equation}%
	\label{eq:PAwithloops}
Pr[u = v_i] = \left\{  \begin{array}{ ll } 
	\frac{d_t(v_i)}{2t-1}, & \text{if} \,\, v_i \neq v_t \\
	\frac{1}{2t-1}, & \text{if}\,\, v_i = v_t
	\end{array} \right.  ;
\end{equation}%
And with probability $1-p$ add two new vertices and an edge between them. For some constant $m$, every $m$ steps contract the most recent $m$ vertices added through the PA step to form a super vertex.
Notice that Equation~\eqref{eq:PAwithloops} is not quite the same as $\gamma_t$ in Equation~\eqref{eqn:gamma}. Equation~\eqref{eq:PAwithloops} allows for nodes to be added with self loops. In both versions loops are allowed in the edge step. Regardless, the allowance of self loops has little effect as the graph becomes large, and we remove all self-loops in our final graph for experimental analysis. 

%Note that this model could be generalized. With probability $1-p$ we could do something completely different with the graph. We could add a clique, or other motif, and it doesn't change the analysis, as long as the degree of each added vertex is less than or equal to $2m$. We discuss potential variations in Section~\ref{sec:GPAResults}.

\subsection{Triad Formation}
\label{sec:Holme}
\citet{holme2002growing} introduced a Triad Formation step into the BA version of the PA model (see Section \ref{sec:PA}). After each PA step in which a new vertex $v$ is added and some edge is added $(v,u)$, a triangle is closed with probability $p_t$ by choosing a neighbor of $u$, $u_2$, and adding edge $(v,u_2)$. An example network is shown in Figure~\ref{fig:examples} under `Holme'. The average number of triad closures per added vertex is $m_t = (m-1)p_t$.  It is shown in~\citet{holme2002growing} that the network follows a power-law in the degrees with an exponent of 3, and has clustering coefficients which can be tuned by the parameter $m_t$. Our model incorporates something very similar to this triad formation, but with less regular structure due to an added component step, and with a larger range of possible power-law exponents. See Sections~\ref{sec:TGPA},~\ref{sec:DegreeTGPA}.

\subsection{Higher Order Features in Graphs}
\label{sec:HigherOrder}
Recently, there has been interest in analyzing the higher order features in graphs~\cite{yin2017higher, Grilli-2017-higher-order, Rosvall-2014-memory, Xu-2016-higher-order, benson2016higher, Benson-2017-srw}. One of the earlier motivations for this direction is the famous paper by~\citet{milo_network_motifs} on the presence of motifs in real world networks.  Likewise, there are new models which aim to match these higher order features. For example the triad formation model described in Section~\ref{sec:Holme}~\cite{holme2002growing}, and the family of PA models~\cite{ostroumova2013generalized} discussed in Section~\ref{sec:Intro}. Another model, HyperKron, places a distribution over hyperedges and inserts motifs instead of edges~\cite{Eikmeier-preprint-HyperKron} and is specifically shown to have higher order clustering.

\section{Eigenvalue Power-law in GPA}
\label{sec:GPAResults}
In this section, we present results for the Generalized Preferential Attachment model presented in~\citet{avin2017improved} and discussed in Section~\ref{sec:GPA}, relating to the distribution of the eigenvalues of a graph formed in the model. Note that in order to get our desired result (Theorem~\ref{Theorem2}), we also prove that the degree distribution has a power-law distribution (Theorem~\ref{Theorem1}). This was already proven in~\citet{avin2017improved}, but the version of our proof is useful in order to obtain Theorem~\ref{Theorem2}. The results and proofs mirror those in~\citet{flaxman2005high}, but provide a useful step towards the results on the TGPA model in Section~\ref{sec:SpectraTGPA}.  Proofs of Lemmas \ref{Lemma1} and \ref{Lemma2}, and Theorem~\ref{Theorem1} are in the supplemental material due to space.

Fix parameter $p$. Denote $G_t^m$ as the Generalized Preferential Attachment Graph at time $t$ with contractions of size $m$.
\begin{lemma}
	\label{Lemma1}
	Let $d_t(s)$ be the degree of vertex $s$ in $G_t^m$, for any time $t$ after $s$ has been added to the graph. Let $a^{(k)} = a(a+1)(a+2)\cdots (a+k-1)$ be the rising factorial function. Let $s'$ be the time at which node $s$ arrives in the graph. Then for any positive integer $k$,
	\[  \mathbb{E}[(d_t(s))^{(k)} ]\leq (2m)^{(k)} 2^{pk/2} \left(  \frac{t}{s'}\right)^{pk/2}    \]
\end{lemma}
%\subsection{Proof of Lemma \ref{Lemma1}}
%\subsection{Proof of Lemma XXX}
\begin{proof}
	Denote $G_t^m$ as the graph at time $t$ with contractions of size $m$. Let $Z_t = d_t^m(s)$ be the degree of vertex $s$ at time $t$, and $Y_t$ an indicator for the event that the edge added at time $t$ is incident to $s$. Then %we can write the expectation of $Z_t$ in terms of $Z_{t-1}$:
\begin{equation}%
\label{eq:expected1}
\begin{aligned}
\mathbb{E}[Z_t^{(k)}] &= \mathbb{E}[ \mathbb{E}[ (Z_{t-1} + Y_t)^{(k)} ]  | Z_{t-1}   ]   \\ &= \mathbb{E}\left[ Z_{t-1}^{(k)}\left(1 - p\cdot \frac{Z_{t-1}}{2t-1} \right) + (Z_{t-1} + 1)^{(k)} \left(p\cdot \frac{Z_{t-1}}{2t-1} \right)      \right]  \\
&= \mathbb{E} [ Z_{t-1}^{(k)}] \left(1 + \frac{pk}{2t-1} \right). \nonumber
\end{aligned}
\end{equation}%
%where the last equality is an easy algebra step. 
Apply this relationship iteratively, down to the time when node $s$ was added (denoted as $s'$). Also note that the degree of $s$ at time $s'$ is bounded by $2m$ (if all m edges were added as self loops).
\begin{equation}%
\begin{aligned}
\mathbb{E}(Z_t^{(k)}) &= \prod_{t' = s'}^{t} \left(  1 + \frac{pk}{2t' - 1} \right) \leq (2m)^{(k)} \prod_{t' = s'+1}^{t}  \left(  1 + \frac{pk}{2t' - 1} \right) %\\ 
%&\leq (2m)^{(k)} \text{exp}\left\{ \frac{pk}{2} \sum_{t' = s+1}^{t} \frac{1}{t'- 1/2}  \right\} \\
\nonumber
\end{aligned}
\end{equation}%
Use $1+x \leq e^x$ to write the product as a sum, and bound
\begin{equation}%
\begin{aligned}
\sum_{t' = s'+1}^t \frac{1}{t'-1/2} \leq \int_{x = s'}^t \frac{1}{x-1/2} \, \, \text{dx} = \log \frac{t -1/2}{s'-1/2}. \hspace{0.1in} \text{\normalsize Then}
\nonumber
\end{aligned}
\end{equation}%
%So finally,
\begin{equation}%
\begin{aligned}
\mathbb{E}(Z_t^{(k)}) &\leq (2m)^{(k)} \left( \frac{t -1/2}{s' -1/2} \right)^{pk/2} = (2m)^{(k)} \left(  \frac{t}{s'} \right)^{pk/2}  \left( \frac{2 - 1/t}{2 - 1/s'} \right)^{pk/2} \\
&\leq (2m)^{(k)}  \left(  \frac{t}{s'} \right)^{pk/2} 2^{pk/2}. \hspace{1.25in} 
\nonumber
\end{aligned}
\end{equation}%
\end{proof}%

Now define a \emph{supernode} to be a collection of nodes viewed as one. The degree of a supernode is the sum of the degrees of the vertices in the supernode.
\begin{lemma}
	\label{Lemma2}
	Let $S = (S_1, S_2, \ldots, S_l)$ be a disjoint collection of supernodes at time $t_0$. Assume that the degree of $S_i$ at time $t_0$ is $d_{t_0} (S_i) = d_i$. Let $t$ be a time later than $t_0$. Let $p_S(\vr;\vd,t_0,t)$ be the probability that each supernode $S_i$ has degree $r_i + d_i$ at time $t$. Let $d = \sum_{i = 1}^l d_i, r = \sum_{i = 1}^{l} r_i$. If $d = o(t^{1/2})$ and $r = o(t^{2/3})$, then
		\begin{equation}%
		\begin{aligned}
		p_S( \vr ; \vd , t_0, t) \leq \left( \prod_{i = 1}^l {r_i + d_i-1 \choose d_i -1} \right) \left( \frac{t_0 + 1}{t} \right)^{pd/2}  \text{exp} \left\{2 + t_0 - \frac{pd}{2} + \frac{3pr}{t^{p/2}} \right\} \nonumber
		\end{aligned}
		\end{equation}
\end{lemma}
%\subsection{Proof of Lemma \ref{Lemma2}}
%\subsection{Proof of Lemma XXX}
\begin{proof}
	Let $\tau^{(i)} = (\tau_1^{(i)}, \ldots, \tau_{r_i}^{(i)})$, where $\tau_j^{(i)}$ is the time when we add an edge incident to $S_i$ and increase the degree from $d_i + j -1$ to $d_i + j$. Define $\tau = (\tau_0, \tau_1, \ldots, \tau_{r+1})$ to be the ordered union of $\tau^{(i)}$, with $\tau_0 = t_0$ and $\tau_{r+1} = t$. Let $p(\tau; \vd, t_0, t)$ be the probability that $S_i$ increases in degree at exactly the times specified by $\tau$. %between time $t_0$ and $t$.
	\begin{fullwidth}
	\begin{equation}%
	\label{eqn:first}
	\begin{aligned}
	p(\tau; \vd, t_0, t)
	&= \left( \prod_{i = 1}^l \prod_{k=1}^{r_i} p \frac{d_i + k - 1}{2\tau_{k}^{(i)} -1} \right)  
	\left( \prod_{k=0}^r \prod_{j = \tau_k + 1}^{\tau_{k+1}-1} \left(1 - p \frac{d+k}{2j-1} \right) \right) \\
		%= \underset{\parbox{70pt}{\scriptsize for each supernode $S_i$, the  prob. of $\tau$ aligning with $\tau^{(i)} $. }}{\left( \prod_{i = 1}^l \prod_{k=1}^{r_i} p \frac{d_i + k - 1}{2\tau_{k}^{(i)} -1} \right)}  \hspace{0.1in} \underset{ \parbox{90pt}{\scriptsize for each timestep inbetween the relevant ones, the prob. of picking any edge outside of $ S_1, \ldots, S_l $.}}{\left( \prod_{k=0}^r \prod_{j = \tau_k + 1}^{\tau_{k+1}-1} \left(1 - p \frac{d+k}{2j-1} \right) \right)} \\
	 &= \left( \prod_{i=1}^l \frac{(r_i + d_i-1)!}{(d_i-1)!} \right) \left( \prod_{k=1}^r \frac{p}{2 \tau_k - 1} \right)  
	 \text{exp} \left\{ \sum_{k = 0}^r \sum_{j = \tau_k +1}^{\tau_{k+1} - 1} \log \left( 1 - p \left( \frac{d+k}{2j-1} \right) \right) \right\}
	 \nonumber
	\end{aligned}
	\end{equation}%
	\end{fullwidth}
	Now we can bound the inner most sum of the exponential term.
\begin{equation}%
\sum_{j = \tau_k + 1}^{\tau_{k+1} - 1} \log \left( 1 - p \left(\frac{d+k}{2j-1} \right) \right) \leq \sum_{j = \tau_{k} + 1}^{\tau_{k+1}-1} \log \left( 1 - \frac{p(d+k)}{2j} \right)
\nonumber
\end{equation}%
which is less than or equal to
\begin{equation}%
\begin{aligned}
\int_{\tau_k + 1}^{\tau_{k+1}} \log \left( 1 - \frac{p(d+k)}{2x} \right) \, \, \text{dx} 
	 & = - \tau_{k+1} \log(2 \tau_{k+1}) + (\tau_k + 1)\log(2 \tau_k + 2) \\
	& -1/2(2 \tau_{k+1} - p(d+k))\log(2 \tau_{k+1} - p(d+k))  \\
	&- 1/2(2 \tau_k + 2 - p(d+k)) \log(2 \tau_k + 2 - p(d+k)).
	\nonumber
\end{aligned}
\end{equation}%
	Note that $\tau_0 = t_0$ and $\tau_{r+1} = t$. We can write
	\begin{equation}% 
	\sum_{k =0}^r \int_{\tau_k + 1}^{\tau_{k+1}} \log \left( 1 - \frac{p(d+k)}{2x} \right) \, \, \text{dx} = A + \sum_{k = 1}^{r} B_k
	\nonumber
	 \end{equation}%
	 where
	 \begin{equation}%
	 \label{eqn:A}
	 \begin{aligned}
	 A &= (t_0 + 1)\log(2t_0 + 2) -1/2(2 t_0 + 2 - pd)\log(2t_0 + 2 - pd) \\ & - t \log(2t) + 1/2 (2t - p(d+r)) \log(2t - p(d+r))
	 \end{aligned}
	 \end{equation}%
	 and
	 \begin{equation}%
	 \label{eqn:Bk}
	 \begin{aligned}
	 B_k &= \tau_k \log(1 + 1/\tau_k) + \log(2 \tau_k + 2) - \frac{2-p}{2} \log(2\tau_k + p - p(d+k))
	 \\&+ 1/2(2 \tau_k + 2 - p(d+k)) \log\left(1 - \frac{2-p}{2\tau_k + 2 -p(d+k)}\right).
	 \end{aligned}
	 \end{equation}%
	 Bound each of $A$ and $B_k$, starting with $B_k$. Since $1+x \leq e^x$, $\tau_k \log(1 + 1/\tau_k) \leq 1$, and $ \frac{1}{2}(2 \tau_k + 2 - p(d+k)) \log\left(1 - \frac{2-p}{2\tau_k + 2 -p(d+k)}\right) \leq -p/2$. Rearranging the other two terms of Equation \eqref{eqn:Bk} we get
	 \begin{equation}%
	 B_k \leq \frac{p}{2} \log(2 \tau_k + 2) - \frac{2-p}{2} \log\left(1 - \frac{p(d+k) + 2-p}{2 \tau_k + 2} \right) + \frac{p}{2}.
	 \nonumber
	 \end{equation}%
Rearranging terms of $A$ from Equation~\eqref{eqn:A} and taking the exponential, 
	 \begin{equation}%
	 \begin{aligned}
	% &A = -(t_0 + 1)\log\left(1 - \frac{pd}{2t_0 + 2} \right) + \frac{pd}{2} \log\left(2t_0 + 2 - pd\right) \\&+ t\log\left( 1 - \frac{p(d+r)}{2t} \right) - \frac{p(d+r)}{2} \log(2t - p(d+r)) \\
	 e^A &= \left( 1 - \frac{pd}{2t_0 + 2} \right)^{-(t_0 + 1)} (2t_0 + 2 - pd)^{pd/2} \left(1 - \frac{p(d+r)}{2t} \right)^t
	  (2t - p(d+r))^{\frac{-p(d+r)}{2}} \\
	 & = \left( 1 - \frac{pd}{2t_0 +2} \right)^{-(1-\frac{pd}{2(t_0+1)})(t_0+1)} \left( 1 - \frac{p(d+r)}{2t} \right)^{t - \frac{p(d+r)}{2}}  
 \left( \frac{t_0 +1}{t} \right)^{\frac{pd}{2}} (2t)^{\frac{-pr}{2}}.
 \nonumber
	 \end{aligned}
	 \end{equation}%
	 Using the bound $1 - x \leq e^{-x - x^2/2}$ for $0 < x < 1$,
	 \begin{equation}%
	 \begin{aligned}
	 \left( 1 - \frac{p(d+r)}{2t} \right)^{t - p(d+r)/2} \leq \text{exp} \left\{ -\frac{p(d+r)}{2} + \frac{p^2(d+r)^2}{8t} + \frac{p^3(d+r)^3}{16t^2}  \right\}
	 \nonumber
	 \end{aligned}
	 \end{equation}%
	 Putting the bounds on $A$ and $B_k$ together, we get
	 \begin{fullwidth}
	 \begin{equation}%
	 \label{eqn:eAB}
	 \begin{aligned}
	 e^{A + \sum B_k} &\leq \left( 1 - \frac{pd}{2t_0 +2} \right)^{-(1-\frac{pd}{2(t_0+1)})(t_0 + 1)}
	  \text{exp} \left\{ -\frac{p(d+r)}{2} + \frac{p^2(d+r)^2}{8t} + \frac{p^3(d+r)^3}{16t^2}  \right\} 
	 \\& \times \left( \frac{t_0 +1}{t} \right)^{pd/2} (2t)^{-pr/2}   \prod_{k=1}^r \left( \left(1 - \frac{p(d+k) + 2-p}{2 \tau_k +2} \right)^{-(2-p)/2} (2 \tau_k + 2)^{p/2} \right) e^{pr/2}.
	 \end{aligned}
	 \end{equation}%
	 \end{fullwidth}
	Define err$(r,d,t_0,t)$
	\begin{fullwidth}
	\begin{equation}%
	\begin{aligned}%
	\text{err}(r,d,t_0,t) &= \left( 1 - \frac{pd}{2t_0 +2} \right)^{-(1-\frac{pd}{2(t_0+1)})(t_0+1)} 
	\text{exp} \left\{ -\frac{pd}{2} + \frac{p^2(d+r)^2}{8t} + \frac{p^3(d+r)^3}{16t^2}  \right\},
	\nonumber
	\end{aligned}%
	\end{equation}%
	\end{fullwidth}
Then we can re-write Equation~\eqref{eqn:eAB} as
	\begin{equation}
	 \text{err}(r,d,t_0,t) \left( \frac{t_0+1}{t} \right)^{\frac{pd}{2}} (2t)^{\frac{-pr}{2}} 
	     \prod_{k=1}^r \left( \left(1 - \frac{p(d+k) + 2-p}{2 \tau_k +2} \right)^{\frac{-(2-p)}{2}} (2 \tau_k + 2)^{\frac{p}{2}} \right).
	     \nonumber
 	\end{equation}%

% 	So we finally finish with the bound on $p(\tau; \vd, t_0,t)$ by substituting Equation \ref{eqn:eAB} into Equation \ref{eqn:first}:
% 	\begin{smalleq}%
% 	\begin{equation}%
% 	\begin{aligned}
% 	p(\tau; \vd, t_0,t) &\leq \left( \prod_{i=1}^l \frac{(r_i + d_i -1)!}{(d_i-1)!} \right) \, \text{err}(r,d,t_0,t) \left( \frac{t_0 + 1}{t} \right)^{\frac{pd}{2}} (2t)^{\frac{-pr}{2}} \\
% 	 &\times \prod_{k=1}^r \left( \left( 1 - \frac{p(d+k) + 2-p}{2 \tau_k + 2} \right)^{\frac{-(2-p)}{2}} (2 \tau_k +2)^{\frac{p}{2}} \frac{p}{2\tau_k - 1} \right), 
% 	 \end{aligned}
% 	 \end{equation}%
% 	\end{smalleq}%
% which can be re-arranged as
% \begin{smalleq}%
% 	\begin{equation}%
% 	\begin{aligned}
% 	&= \left( \prod_{i=1}^l \frac{(r_i + d_i -1)!}{(d_i-1)!} \right) \, \text{err}(r,d,t_0,t)\left( \frac{t_0 + 1}{t} \right)^{pd/2} (2t)^{-pr/2} \\ 
% 	& \times  \prod_{k=1}^r \left( p(2\tau_k +p - p(d+k))^{-(2-p)/2} \left(1 + \frac{3}{2 \tau_k -1} \right) \right).
% 	\end{aligned}
% 	\end{equation}%
% \end{smalleq}%

So we finally finish with the bound on $p(\tau; \vd, t_0,t)$ by substituting Equation~\eqref{eqn:eAB} into Equation~\eqref{eqn:first} and rearranging terms:
	\begin{equation}%
	\begin{aligned}
p(\tau; \vd, t_0,t)	&= \left( \prod_{i=1}^l \frac{(r_i + d_i -1)!}{(d_i-1)!} \right) \, \text{err}(r,d,t_0,t)\left( \frac{t_0 + 1}{t} \right)^{pd/2} (2t)^{-pr/2} \\ 
	& \times  \prod_{k=1}^r \left( p(2\tau_k +p - p(d+k))^{-(2-p)/2} \left(1 + \frac{3}{2 \tau_k -1} \right) \right).
	\nonumber
	\end{aligned}
	\end{equation}%
Now, we will sum $p(\tau; \vd, t_0,t)$ over all ordered choices of $\tau$.
 	\begin{equation}%
 	\begin{aligned}
 &	p(\vr; \vd, t_0, t) \leq \sum_{\tau^{(1)}, \ldots, \tau^{(l)}} p(\tau; \vd, t_0,t) 
 \nonumber
 \end{aligned}
 \end{equation}%
	\begin{equation}%
\begin{aligned}
 		&\leq {r \choose r_1, \ldots, r_l} \sum_{t_0 + 1 \leq \tau_1 < \ldots < \tau_r \leq t}  \prod_{i=1}^l \frac{(r_i + d_i -1)!}{(d_i-1)!} 
 		\text{err}(r,d,t_0,t) \left( \frac{t_0 + 1}{t} \right)^{\frac{pd}{2}} \\
 		& \hspace{0.5in} \times  (2t)^{\frac{-pr}{2}}p \prod_{k=1}^r (2\tau_k +p - p(d+k))^{-(2-p)/2} \left(1 + \frac{3}{2 \tau_k -1} \right)  \nonumber
 		\end{aligned}
 		\end{equation}%
\begin{equation}%
\begin{aligned}
 	\label{eqn:almostfinal}
 		&= r! \left( \prod_{i=1}^l {r_i + d_i - 1 \choose d_i - 1} \right) \, \text{err}(r,d,t_0,t) \left(\frac{t_0 + 1}{t}\right)^{pd/2} (2t)^{-pr/2} \\
 		&\hspace{0.2in} \times \sum_{t_0 + 1 \leq \tau_1 < \ldots < \tau_r \leq t} p \prod_{k=1}^r (2 \tau_k + p - p(d+k))^{-(2-p)/2} \left(1 + \frac{3}{2 \tau_k - 1} \right) 
 	\end{aligned}
 	\end{equation}%
Now let $\tau_k' = \tau_k - \lceil p(d+k)/2 \rceil$. Since $d \geq 1$ and $k \geq 1$, we have $2 \lceil p(d+k)/2 \rceil \geq 2$. So, Equation~\eqref{eqn:almostfinal} is less than or equal to
\begin{fullwidth}
 	\begin{equation}%
 	\begin{aligned}
 	\sum_{ (t_0 - p \lceil d/2 \rceil + 1) \leq \tau_1' \leq \ldots \leq \tau_r' \leq (t - p \lceil (d+r)/2 \rceil ) } &\left( p \prod_{k = 1}^r (2 \tau_k' + p)^{-(2-p)/2} \left( 1 + \frac{3}{2\tau_k' +1} \right) \right)  \\
 	&\leq \frac{p}{r!} \left( \sum_{ \tau' = (t_0 - p \lceil d/2 \rceil + 1)}^{t - p \lceil (d+r)/2 \rceil} \left( 2 \tau' + p)^{-(2-p)/2} + 3(2\tau' + 1)^{-(4-p)/2} \right) \right)^r \\
 	%& \leq \frac{p}{r!} \left( \int_{ 0}^{t - p (d+r)/2 } \left(   (2x + p)^{-(2-p)/2} + 3(2x + 1)^{-(4-p)/2} \right) \text{dx} \right)^r \\
 	%& \leq \frac{p}{r!} \left(  \frac{1}{p}(2t - p(d+r) + p)^{p/2} - p^{(p-2)/2} - \frac{3}{2-p}(2t - p(d+r) + 1)^{-(2-p)/2} + \frac{3}{2-p}   \right)^r \\
 	%& \leq \frac{p}{r!} \left( \frac{1}{p}(2t - p(d+r) + p)^{p/2} + 3   \right)^r \\
 	%& =  \frac{p}{r!} \left( \frac{1}{p} (2t)^{p/2} \left(1- \frac{p(d+r) - p}{2t}\right)^{p/2} \left( 1 + \frac{3p}{(2t - p(d+r) + p)^{p/2}} \right) \right)^r \\
 	&= \frac{1}{r!} (2t)^{pr/2} \underbrace{\left(1 - \frac{p(d+r) - p}{2t}  \right)^{pr/2}}_{\leq \text{exp}\left\{-\frac{rp(p(d+r) -p)}{4t} \right\} } \underbrace{ \left( 1 + \frac{3p}{(2t - p(d+r) +p)^{p/2}} \right)^r}_{\leq \text{exp} \left\{  \frac{3pr}{ (2t - p(d+r) + p)^{p/2} }   \right\}  } .
 	\nonumber
 	\end{aligned}
 	\end{equation}%
 	\end{fullwidth}
 	Where the last inequalities come from $1 + x \leq e^x$. So finally, 
 	\begin{fullwidth}
 	\begin{equation*}%
 	\begin{aligned}
 	p_S(\vr; \vd, t_0, t) &\leq \left( \prod_{i = 1}^l {r_i + d_i - 1 \choose d_i - 1} \right) \text{err}(r,d,t_0,t) \left( \frac{t_0 + 1}{t} \right)^{pd/2} 
 	 \text{exp} \left \{  \frac{-rp((d+r) -p)}{4t} + \frac{3pr}{(2t - p(d+r) + p)^{p/2}}  \right\}.
 	\nonumber
 	\end{aligned}
 	\end{equation*}%
 	\end{fullwidth}
% 	Since $d = o(t^{1/2})$ and $r = o(t^{2/3})$, 
% 	\begin{smalleq}%
% 	\begin{equation*}%
% 	\begin{aligned}
% 	&\text{err}(r,d,t_0,t) \text{exp}\left \{  \frac{-rp((d+r) -p)}{4t} + \frac{3pr}{(2t - p(d+r) + p)^{p/2}}  \right\} \\
% 	&\leq \left( 1 - \frac{pd}{2(t_0 + 1)}  \right)^{- (1 + pd/2(t_0 + 1)) (t_0 + 1)} \text{exp} \left\{ 1 - \frac{pd}{2} - \frac{r^2}{8t} + \frac{3pr}{t^{p/2}} \right\} \\
% 	&\leq \underset{\parbox{45pt}{\scriptsize since $x^{-x} \leq e$}}{e^{(t_0+1)}}  \text{exp} \left\{ 1 - \frac{pd}{2} + \frac{3pr}{t^{p/2}} \right\} = \text{exp} \left\{ 2 + t_0 - \frac{pd}{2} + \frac{3pr}{t^{p/2}}   \right\}
% 	\end{aligned}
% 	\end{equation*}%
% \end{smalleq}%
Using $d = o(t^{1/2})$ and $r = o(t^{2/3})$ gives the final bound, and this concludes the proof.
\end{proof}

\begin{theorem}
	\label{Theorem1}
	Let $m$, $k$ be fixed positive integers, and let $f(t)$ be a function with $f(t) \rightarrow \infty$ as $t \rightarrow \infty$. Let $\Delta_1 \geq \Delta_2 \geq \ldots \geq \Delta_k$ denote the degrees of the $k$ highest degree vertices of $G_t^m$. Then
		\begin{equation}%
	 \frac{t^{p/2}}{f(t)} \leq \Delta_1 \leq t^{p/2}f(t) \, \, \, \text{and} \, \, \,
	 \frac{t^{p/2}}{f(t)} \leq \Delta_i \leq \Delta_{i-1} - t^{p/2}f(t)  \nonumber
	 \end{equation}%
	for $i = 1,2,\ldots,k$ whp.
\end{theorem}

%The bounding of the largest degree nodes in Theorem \ref{Theorem1} implies a power-law distribution for the largest degrees.  
The factor of $t^{p/2}$ in Theorem \ref{Theorem1} implies a power-law distribution in the largest degrees with exponent $\beta = (2 + p)/p$. This can be seen by using a martingale argument, as described in~\citet{van2016random} for instance. Notice that depending on the value chosen for $p$, we can obtain a power-law fit with exponents ranging between $3$ and $\infty$. 

%\subsection{Proof of Theorem \ref{Theorem1}}
%\subsection{Proof of Theorem XXX}
\begin{proof}
Partition the vertices into those added before time $t_0$, before time $t_1$, and after $t_1$, with $t_0 = \log \log \log f(t), t_1 = \log \log f(t).$ We will argue about the maximum degree vertices in each set.	
	
\begin{claim}%
		\label{Claim1}
		In $G_t^m$, the degree of the supernode of vertices added before time $t_0$ is at least $t_0^{1/3} t^{p/2}$ whp.
\end{claim}%

	\begin{proof}
		Consider all vertices added before time $t_0$ as a supernode. Let $A_1$ denote the event that this supernode has degree less than $t_0^{1/3}t_{}^{p/2}$ at time $t$. We will use Lemma \ref{Lemma2} with $l = 1$, and $d = 2t_0$. %(because the supernode has all edges at time $t_0$).
		\begin{equation}%
		\begin{aligned}
		Pr[A_1] &\leq \sum_{r_1 = 0}^{t_0^{1/3}t^{p/2}-2t_0} {r_1 + 2t_0 - 1 \choose 2t_0 - 1} \left( \frac{t_0 + 1}{t} \right)^{pd/2} e^{2 + t_0 - pd/2 + 3pr/t^{p/2}} \\
		& \leq \sum_{r_1 = 0}^{t_0^{1/3}t^{p/2} - 2t_0}   \underset{\parbox{55pt}{\scriptsize by  $r_1 = t_0^{1/3}t^{p/2}$ }}{t_0^{1/3}t^{p/2} - 1 \choose 2t_0 - 1} \left( \frac{t_0 + 1}{t} \right)^{pt_0}       
		 \underset{\parbox{75pt}{ \scriptsize by $r = r_1$ and $d = 2t_0$}}{e^{2 + t_0 - pt_0 + 3pt_0^{1/3} - 6pt_0/t^{p/2}}} \nonumber 
		 \end{aligned}
		 \end{equation}%
		
		\begin{equation}%
		\begin{aligned}
		&= (t_0^{1/3}t^{p/2} - 2t_0) \frac{ (t_0^{1/3} t^{p/2} - 1)!  }{  (2t_0-1)! (t_0^{1/3} t^{p/2} - 2t_0)!   }  \left( \frac{t_0 + 1}{t} \right)^{pt_0} 
		 e^{2 + t_0(1-p) + 3pt_0^{1/3} - \frac{6pt_0}{t^{p/2}}} \\
		& \leq t_0^{1/3} t^{p/2} \frac{  (t_0^{1/3}t^{p/2})^{2t_0 - 1}}{(2t_0 - 1)!} \left( \frac{t_0 + 1}{t} \right)^{pt_0} 
	 e^{2 + t_0(1-p) + 3pt_0^{1/3} - 6pt_0/t^{p/2}} \\
		%&= t_0^{2t_0/3} \frac{  1}{(2t_0 -1)!} (t_0 + 1)^{pt_0}e^{2 + t_0(1-p) + 3pt_0^{1/3} - 6pt_0/t^{p/2}} \\
		\nonumber
		\end{aligned}
		\end{equation}%
Using $1/x \leq e^x/x^x$ and rearranging terms, $P[A_1]$ goes to 0:
\begin{equation}
		%&\leq t_0^{2t_0/3} \frac{ e^{2t_0 - 1}}{(2t_0 -1)^{2t_0-1}} (t_0 + 1)^{pt_0} 
		%e^{2 + t_0(1-p) + 3pt_0^{1/3} - 6pt_0/t^{p/2}} \\
		%&\leq \frac{   e^{1 + (3-p)t_0 + 3pt_0^{1/3} - 6pt_0/t^{p/2}} (t_0 + 1)^{(p + 2/3)t_0} }{(2t_0 - 1)^{2t_0 - 1}} \\
Pr[A_1]	\leq \frac{   e^{1 + (3-p)t_0 + 3pt_0^{1/3} - 6pt_0/t^{p/2}} }{(2t_0 - 1)^{t_0(4/3-p) - 1}}.
\nonumber
		\end{equation}%
		\end{proof}

	\begin{claim}
		\label{Claim2}
		In $G_t^m$, no vertex added after time $t_1$ has degree exceeding $t_0^{-2} t^{p/2}$ whp.
	\end{claim}
	\begin{proof}
	Let $A_2$ denote the event that some vertex added after time $t_1$ has degree exceeding $t_0^{-2} t^{p/2}$. 
	\begin{equation}%
	\begin{aligned}
	Pr[A_2] &\leq  \sum_{s = t_1}^{t} \text{Pr}[d_t(s) \geq t_0^{-2} t^{p/2}]
	= \sum_{s = t_1}^{t} \text{Pr}[(d_t(s))^{(l)} \geq (t_0^{-2} t^{p/2})^{(l)}] \\
	&\leq   \sum_{s = t_1}^t t_0^{2l} t^{-lp/2} \mathbb{E}[(d_t(s))^{(l)}] 
	 = \sum_{s = t_1}^{t}t_0^{2l} t^{-lp/2} (2m)^{(l)} 2^{lp/2} \left( \frac{t}{s} \right)^{lp/2}  \nonumber
	 \end{aligned}
	 \end{equation}%
\begin{equation}%
	\label{eqn:A21}
= 2^{lp/2} (2m)^{(l)} t_0^{2l} \int_{t_1-1}^{t} x^{-lp/2} \,\, \text{dx} 
	\end{equation}%
	We compute the integral in Equation~\eqref{eqn:A21}, 
	\begin{equation}%
	\label{eqn:A22}
	\begin{aligned}
\int_{t_1-1}^{t}& x^{-lp/2} \,\, \text{dx} = \left. \frac{x^{-lp/2 + 1}}{-lp/2 + 1} \right|_{t_1 -1}^{t} = (-lp/2 + 1)^{-1} \left( t^{-lp/2+1} - (t_1 - 1)^{-lp/2 + 1} \right) 
	\end{aligned}
	\end{equation}%
	%We want to choose $l$ so that $-lp/2 + 1$ is less than 0. So 
	Choose $l > 2/p$.
	Then the integral in Equation~\eqref{eqn:A22} is less than or equal to $(lp/2 - 1)^{-1} (t_1-1)^{-lp/2 +1}$, and plugging in the computation from Equation~\eqref{eqn:A22} into Equation~\eqref{eqn:A21},
	\begin{equation}
	\begin{aligned}
	Pr[A_2] \leq \frac{ 2^{lp/2} (2m)^{(l)} t_0^{2l} }{(lp/2-1)(t_1-1)^{lp/2 - 1}}
	\nonumber
	\end{aligned}
	\end{equation}
    which goes to 0 as $t$ increases.
	\end{proof}
	
	\begin{claim}
		\label{Claim3}
		In $G_t^m$, no vertex added before time $t_1$ has degree exceeding $t_0^{1/6} t^{p/2}$ whp.
	\end{claim}
\begin{proof}
	Use same technique as in Claim~\ref{Claim2}.
\end{proof}

\begin{claim}
\label{Claim4}
		The k highest degree vertices of $G_t^m$ are added before time $t_1$ and have degree $\Delta_i$ bounded by $t_0^{-1} t^{p/2} \leq \Delta_i \leq t_{0}^{1/6} t^{p/2}_{}$
	\end{claim}
\begin{proof} 
%	First lets summarize the results of the last three claims:
%	\begin{itemize}
%		\item Bound on degrees of vertices added after time $t_1$: $t_0^{-2}t_{}^{p/2}$
%		\item Bound on degrees of vertices added before time $t_1^{}$: $t_0^{1/6} t_{}^{p/2}$
%		\item Sum of all degrees added before time $t_0^{}$ is at least: $t_0^{1/3}t^{p/2}_{}$
%	\end{itemize}
%The upper bound is immediate from Claim~\ref{Claim3}. 
If the lower bound does not hold, then one of the top $k$ vertices has degree less than $t_0^{-1} t^{p/2}$ and the total degree of vertices added before time $t_0$ is bounded by
\begin{equation}%
\begin{aligned}
& (k-1) t_0^{1/6} t^{p/2}  +  \left(\frac{t_0}{m} - k + 1 \right)\left(t_0^{-1} t^{p/2} \right) \leq  kt_0^{1/6} t^{p/2} + t_0(t_0^{-1} t^{p/2}) \\
%&\underbrace{ (k-1) t_0^{1/6} t^{p/2}  }_{\parbox{60pt}{\scriptsize largest possible degrees of $(k-1)$ vertices    }} + \underbrace{ \left(\frac{t_0}{m} - k + 1 \right)\left(t_0^{-1} t^{p/2} \right)  }_{\parbox{80pt}{\scriptsize largest possible degrees of remaning vertices   }} \leq  kt_0^{1/6} t^{p/2} + t_0(t_0^{-1} t^{p/2}) \\
&= k t_0^{1/6} t^{p/2} + t^{p/2} 
= t^{p/2} (k t_0^{1/6} + 1) 
\leq t^{p/2} (2kt_0^{1/6}) 
\leq t^{p/2} t_0^{1/3}. \nonumber
\end{aligned}
\end{equation}%
%But this contradicts Claim \ref{Claim1}. So the lower bound must hold. 
Since we have the lower bound, and we know that $t^{p/2}/t_0 \geq t^{p/2}/t_0^2$, none of the largest degree vertices could be added after time $t_1$.
\end{proof}%
\begin{claim}%
	\label{Claim5}
The $k$ highest degree vertices have 
	$ \Delta_i \leq \Delta_{i-1} - \frac{t^{p/2}}{f(t)}  $ whp.
\end{claim}%
\begin{proof}
	Let $A_4$ denote the event that there are two vertices among the first $t_1$ time steps with degrees exceeding $t_0^{-1}t^{p/2}$ and within $t^{p/2}/f(t)$ of each other. Let $\overline{A}_3$ be the opposite of event $A_3$ from Claim \ref{Claim3}. Let 
		\begin{equation}
		\begin{aligned}
			\label{eqn:A4}
		 p_{l,s_1,s_2} &= \text{Pr}\left[d_t(s_1)-d_t(s_2) = l \, \, | \, \,  \overline{A}_3 \right], \text{for} \, \, |l| \leq t^{p/2}/f(t)
		 \end{aligned}
		 \end{equation}
	 Then
	 \begin{fullwidth}
	 	\begin{equation} 
	 	\begin{aligned} 
	\text{Pr}\left[  A_4 | \overline{A}_3 \right] 
	%\leq   \underset{\parbox{50pt}{\scriptsize for each combination of two vertices}}{ \sum_{1 \leq s_1 < s_2 \leq t_1}} \hspace{.1in} \underset{\parbox{50pt}{\scriptsize for each possible distance $l$}}{ \sum_{l = \frac{ -t^{p/2}}{f(t)}}^{ \frac{t^{p/2}}{f(t)}}} \hspace{.1in} p_{l, s_1, s_2}.
	&\leq \sum_{1 \leq s_1 < s_2 \leq t_1} \sum_{l = -t^{p/2} / f(t)}^{t^{p/2}/f(t)} p_{l,s_1,s_2} \\
	p_{l,s_1,s_2} 
	%&\leq \underset{\parbox{40pt}{\scriptsize all possible degrees for $s_1, s_2$}}{\sum_{r_1 = t_0^{-1}t^{p/2}}^{t_0^{1/6}t^{p/2}}} \hspace{.2in} \underset{\parbox{35pt}{\scriptsize degrees of $s_1, s_2$ at time $t_1$}}{\sum_{d_1, d_2 =1}^{2t_1}} \underset{\parbox{100pt}{\scriptsize Notation from Lemma \ref{Lemma2}. The probability of $(s_1, s_2)$ having degrees $(r_1 + d_1, r_1 - l + d_2)$ at time $t$, where $d_{t_1}(s_i) = d_i$ }}{ p_{(s_1,s_2)} \left(  (r_1, r_1 - l); (d_1, d_2), t_1, t  \right)} \\
	&\leq \sum_{r_1 = t_0^{-1}t^{p/2}}^{t_0^{1/6}t^{p/2}} \sum_{d_1, d_2 =1}^{2t_1} \underset{\parbox{100pt}{\scriptsize Notation from Lemma \ref{Lemma2}. }}{ p_{(s_1,s_2)} \left(  (r_1, r_1 - l); (d_1, d_2), t_1, t  \right)} \\
	&\leq \sum_{r_1 = t_0^{-1}t^{p/2}}^{t_0^{1/6}t^{p/2}} \sum_{d_1, d_2 =1}^{2t_1} { r_1 + d_1 -1 \choose d_1 - 1} {r_1 - l + d_2 - 1 \choose d_2 - 1}  
	\left( \frac{t_1 + 1}{t} \right)^{\frac{p(d_1 + d_2)}{2}}  e^{ \left\{  2 + t_1 - \frac{ p(d_1 + d_2)}{2} + \frac{3p(r_1 - l)}{t^{p/2}}   \right\}} \nonumber \\
	%&\leq t_0^{1/6} t^{p/2} \sum_{d_1, d_2 =1}^{2t_1}  { t_0^{1/6}t^{p/2} + d_1 -1 \choose d_1 - 1} {t_0^{1/6} t^{p/2} - l + d_2 - 1 \choose d_2 - 1}  \\
%	&\hspace{1cm} \times  \left( \frac{t_1 + 1}{t} \right)^{\frac{p(d_1 + d_2)}{2}} e^{ \left\{  2 + t_1 - \frac{ p(d_1 + d_2)}{2} + \frac{3pt_0^{1/6}t^{p/2}}{t^{p/2}}   \right\}} \\
%	&\leq t_0^{1/6} t^{p/2} \sum_{d_1, d_2 =1}^{2t_1}  { 2t_0^{1/6}t^{p/2} \choose d_1 - 1} {2 t_0^{1/6} t^{p/2} \choose d_2 - 1}   \left( \frac{t_1 + 1}{t} \right)^{\frac{p(d_1 + d_2)}{2}}  e^{  2 + t_1  +3pt_0^{1/6}}  \\
	%&\leq t_0^{1/6} t^{p/2} \sum_{d_1, d_2 = 1}^{2t_1} \underset{\parbox{60pt}{\scriptsize because ${x \choose y} \leq x^y$}}{(2t_0^{1/6} t^{p/2} )^{d_1 + d_2 -2} } \underset{\parbox{40pt}{\scriptsize because $d_1 + d_2 \leq 4t_1$}}{ (t_1 + 1)^{2pt_1}}  t^{-p(d_1 + d_2)/2} e^{2 + t_1  +3pt_0^{1/6}} \\
	&\leq t_0^{1/6} t^{p/2} \sum_{d_1, d_2 = 1}^{2t_1} (2t_0^{1/6} t^{p/2} )^{d_1 + d_2 -2}   (t_1 + 1)^{2pt_1}}  t^{-p(d_1 + d_2)/2} e^{2 + t_1  +3pt_0^{1/6} \\	
	%&\leq  t_0^{1/6} t^{p/2} (2t_1)^2 2^{4t_1} t_0^{2t_1/3} t^{p/2(d_1 + d_2 - 2)} (t_1+1)^{2pt_1} t^{ - p(d_1 + d_2)/2} e^{2 + t_1  +3pt_0^{1/6}} \\
	&= t^{-p/2} t_0^{1/6} (2t_1)^2 2^{4t_1} t_0^{2t_1/3} (t_1+1)^{2pt_1} e^{ 2 + t_1  +3pt_0^{1/6}} \nonumber
	%&=: h(t)
	\end{aligned}
	\end{equation}%
	\end{fullwidth}
	Denote the last equation as $h(t)$ and note $h(t)$ is a polynomial in $\log(f(t))$ times a factor of $t^{-p/2}$. Then going back to Equation~\eqref{eqn:A4},
	\begin{equation}%
	\begin{aligned}
	\text{Pr}\left[  A_4 | \overline{A}_3 \right] 
	%\leq \sum_{1 \leq s_1 < s_2 \leq t_1}  \sum_{l = \frac{ -t^{p/2}}{f(t)}}^{ \frac{t^{p/2}}{f(t)}}  p_{l, s_1, s_2}  
	\leq {t_1 \choose 2} 2 \frac{t^{p/2}}{f(t)} h(t) 
	= {t_1 \choose 2} 2\frac{\text{poly}(\log(f(t)))}{f(t)} \nonumber
	\end{aligned}
	\end{equation}%
	which goes to $0$ as $t$ increases.
\end{proof}%
Finishing that final Claim finishes the proof of the theorem.
\end{proof}%

The next result relates maximum eigenvalues and maximal degrees in the GPA model. It is similar to results found in~\citet{MihailPapadimitriou2002,ChungLuVu2003,chung2003spectra,flaxman2005high}. It says that if the degrees follows a power-law with exponent $\beta$, then the spectra follows a power-law as well.
\begin{theorem}
	\label{Theorem2}
	Let $k$ be a fixed integer, and let $f(t)$ be a function with $f(t) \rightarrow \infty$ as $t \rightarrow \infty$. Let $\lambda_1 \geq \lambda_2 \geq \ldots \geq \lambda_k$ be the $k$ largest eigenvalues of the adjacency matrix of $G_t^m$. The for $i = 1, \ldots, k$, we have $\lambda_i = (1 + o(1))\Delta_i^{1/2}$, where $\Delta_i$ is the $i$\nth largest degree.
\end{theorem}

%\subsection{Proof of Theorem \ref{Theorem2} }
%\subsection{Proof of Theorem XXX }
\begin{proof}
Let $G = G_t^m$. We will show that with high probability $G$ contains a star forest $F$, with stars of degree asymptotic to the maximum degree vertices of $G$. Then show that $G \backslash F$ has small eigenvalues. Then we can use Rayleigh's principle to say that the large eigenvalues of $G$ cannot be too different than the large eigenvalues of $F$. 

Let $S_i$ be the vertices added after time $t_{i-1}$ and at or before time $t_i$, for $t_0 = 0, t_1 = t^{1/8}, t_2 = t^{9/16}, t_3 = t $.
We start by finding bounds on the degrees of $G$.
\begin{claim}%
	\label{Claim6}
For any $\varepsilon > 0$, and any $f(t)$ with $f(t) \rightarrow \infty$ as $t \rightarrow \infty$ the following holds whp: for all $s$ with $f(t) \leq s \leq t$, for all vertices $v \in G_s$, if $v$ was added at time $r$, then $d_s(v) \leq s^{p/2 + \varepsilon} r^{-p/2}$.
\end{claim}%
\begin{proof}
	\begin{equation}%
	\begin{aligned}
	\text{Pr} &\left[ \cup_{s = f(t)}^t \cup_{r=1}^{s} \left\{  d_s^m(r) \geq s^{p/2 + \varepsilon} r^{-p/2}   \right\}   \right] 
	 \leq \sum_{s = f(t)}^t \sum_{r=1}^{s}  \text{Pr}\left[   d_s^m(r) \geq s^{p/2 + \varepsilon} r^{-p/2}     \right] \\
	&= \sum_{s = f(t)}^t \sum_{r=1}^{s}  \text{Pr}\left[   (d_s^m(r))^{(l)} \geq (s^{p/2 + \varepsilon} r^{-p/2})^{(l)}     \right] \nonumber
	\end{aligned}
	\end{equation}%
which is bounded using Markov:
\[\leq \sum_{s = f(t)}^{t} \sum_{r = 1}^{s} s^{-l(p/2 + \varepsilon)} r^{pl/2} \mathbb{E}\left[  (d_s^m(r))^{(l)} \right] \] \nonumber
which we can bound using Lemma \ref{Lemma1}
\[ \leq \sum_{s = f(t)}^{t} \sum_{r = 1}^{s} s^{-l(p/2 + \varepsilon)} r^{lp/2} (2m)^{(l)} 2^{lp/2} \left( \frac{s}{r}  \right)^{lp/2}
	= (2m)^{(l)} 2^{lp/2} \sum_{s = f(t)}^{t} s^{1-\varepsilon l}   \]
	Take $l \geq 3/\varepsilon$. Then we can bound the sum by an integral,
	\begin{equation}%
	\begin{aligned}
	\sum_{s = f(t)}^t s^{1-\varepsilon l} &\leq \int_{f(t) - 1}^{\infty} x^{1-\varepsilon l} \, \, \text{dx} 
	= \left. \frac{1}{2 - \varepsilon l } x^{2 - \varepsilon l } \right|_{f(t) -1}^{\infty} 
	%= \frac{1}{2 - \varepsilon l } \left( 0 - (f(t) -1)^{2 - \varepsilon l}  \right) 
	=\frac{1}{ \varepsilon l - 2 } (f(t) -1)^{2 - \varepsilon l} \nonumber
	\end{aligned}
	\end{equation}%
	which goes to zero as $t$ increases, since $l \geq 3/\varepsilon$.
\end{proof}%

\begin{claim}
	\label{Claim7}
	Let $S_3'$ be the set of vertices in $S_3$ that are adjacent to more than one vertex of $S_1$ in $G$. Then $|S_3'| \leq t^{7p/16}$ with high probability.
\end{claim}
\begin{proof}
	Let $B_1$ be the event that the conditions of Claim \ref{Claim6} hold with $f(t) = t_2$ and $\varepsilon = 1/16$. Then for a vertex $v \in S_3$ added at time $s$, the probability that $v$ picks at least one neighbor in $S_1$ is less than or equal to 
	\begin{equation}%
	\frac{ \sum_{w \in S_1} d_s(w) }{2s-1}  \leq \frac{ \sum_{w \in S_1} s^{p/2 + \varepsilon} }{2s-1} = \frac{t_1 s^{p/2 + \varepsilon}}{2s-1} \nonumber
	\end{equation}%
	So, the probability of having two or more neighbors in $S_1$ can be bounded by,
	\begin{equation}%
	\begin{aligned}
	\text{Pr}\left[ \left. |N(v) \cap S_1 | \geq 2 \, \, \right|  B_1 \right] &\leq \left( \frac{t_1 s^{p/2 + \varepsilon} }{2s-1} \right)^2 \cdot \binom{m}{2} 
	\leq m^2 t^{1/4}s^{(-15 + 8p)/8} \nonumber
	\end{aligned}
	\end{equation}%
	Let $X$ denote the number of $v \in S_3$ adjacent to more than one vertex of $S_1$. Then 
	\begin{equation}%
	\begin{aligned}
	\mathbb{E}[ X | B_1 ] &\leq \sum_{t_2 + 1}^{t} m^2 s^{(-15 + 8p)/8}t^{1/4} 
	\leq m^2 t^{1/4} \int_{t_2}^t x^{(-15 + 8p)/8} \,\, \text{dx} 
	\\&= m^2 t^{1/4} \left[ \left. \frac{8}{7 + 8p} x^{(-7 + 8p)/8} \right|_{t_2}^t \right] 
	\leq \frac{8 m^2 t^{1/4}}{-7 + 8p} t^{(-7 + 8p)/8}  \nonumber
	\end{aligned}
	\end{equation}%
	Then by Markov,
	\begin{equation}%
	\begin{aligned}
	\text{Pr}&\left[ X \geq t^{7p/16} | B_1 \right] \leq \frac{ \mathbb{E}[X | B_1] }{t^{7p/16}} 
	\leq \frac{ 8m^2}{8p - 7} \frac{ t^{p - 5/8}}{t^{7p/16}} 
	%\\& = \frac{ 8m^2}{8p - 7} \frac{ t^{p} t^{-5/8}}{t^{p}t^{-9p/16}} 
	= \frac{ 8m^2}{8p - 7} \frac{ t^{-5/8}}{t^{-9p/16}} \nonumber
	\end{aligned}
	\end{equation}%
	And $ \frac{ t^{-5/8}}{t^{-9p/16}}  = \frac{t^{9p/16}}{t^{5/8}} \leq \frac{t^{9/16}}{t^{5/8}}$ which goes to zero. 
	\end{proof}

Let $F \subseteq G$ be the star forest consisting of edges between $S_1$ and $S_3 \backslash S_3'$.

	\begin{claim} %
		\label{Claim8}
		Let $\Delta_1 \geq \Delta_2 \geq \ldots \geq \Delta_k$ denote the degrees of the $k$ highest degree vertices of $G$. Then $\lambda_i(F) = (1 - o(1)) \Delta_i^{1/2}$.
	\end{claim}%
	\begin{proof}
		Denote $K_{1,d_i}$ to be a star of degree $d_i$. Let $H$ be the star forest $H = K_{1,d_1} \cup \ldots \cup K_{1,d_k}$ with $d_1 \geq d_2 \geq \ldots \geq d_k$. Then for $i = 1, \ldots, k$, $\lambda_i(H) = d_i^{1/2}$. So it will be sufficient to show that $\Delta_i(F) = (1-o(1)) \Delta_i(G)$. Within the proof of Theorem~\ref{Theorem1}, we show that the $k$ highest degree vertices $G$ are added before time $t_1$ (specifically in Claim \ref{Claim4} in the Supplement). So these vertices are all in $F$. The only edges to those vertices that are not in $F$ are those added before time $t_2$ and those incident to $S_3'$. 
		
		By Theorem \ref{Theorem1} we can choose $f(t)$ such that, $\Delta_1(G_{t_2}^m) \leq t_2^{p/2} f(t) \leq t^{7p/16}$. Also by Theorem \ref{Theorem1}, $\Delta_i(G) \geq t^{p/2} / \log{t}$. Finally, Claim \ref{Claim7} says that $|S_3'| \leq t^{7p/16}$ whp. And so, with high probability,
		\begin{equation}%
		\begin{aligned}
		&\Delta_i(F)  \geq \Delta_i(G) - t^{7p/16} - mt^{7p/16}
		\geq \frac{t^{p/2}}{\log{t}} - t^{7p/16} (1 + m) \\
		& = \frac{t^{p/2}}{\log{t}} \left[ 1 - t^{7p/16} (1+m) \frac{ \log{t}}{t^{p/2}}   \right] 
		= \frac{t^{p/2}}{\log{t}} \left[ 1 - (1+m) \frac{ \log{t}}{t^{p/2-7p/16}}   \right] 
		&\\& = \frac{t^{p/2}}{\log{t}} \left[ 1 - (1+m) \frac{ \log{t}}{t^{p/16}}   \right] 
		= (1-o(1)) \Delta_i(G) \nonumber
		\end{aligned}%
		\end{equation}%
	\end{proof}

Let $H = G \backslash F$. Denote $\mA_G, \mA_F$ and $\mA_H$ to be the adjacency matrices for graphs $G,F$ and $H$. In the following claim, we'll show that $\lambda_1(\mA_H)$ is $o(\lambda_k(\mA_F))$. Consider the fact that if $\mA$ and $\mA+\mE$ are symmetric $n \times n$ matrices, then $\lambda_k(\mA) + \lambda_n(\mE) \leq \lambda_k(\mA) + \lambda_1(\mE)$ (see for instance~\citet{GolubVanLoan2013}). That implies that for any subspace $L$, 
\[   \underset{\vx \in L,\vx \neq 0}{\text{max}} \frac{\vx^T \mA_G \vx}{\vx^T\vx} = \underset{\vx \in L, \vx \neq 0}{\text{max}}  \frac{\vx^T \mA_F \vx}{\vx^T \vx} \pm O \left(   \underset{\vx \neq 0}{\text{max}} \frac{\vx^T \mA_H \vx}{\vx^T \vx}  \right). \]
This is enough to finish the proof because by Rayleigh's Principle~\cite{GolubVanLoan2013},
$ \lambda_i(\mA_G) = \lambda_i(\mA_F)(1 \pm o(1) )$.

\begin{claim}
	\label{Claim9}
	$\lambda_1(\mA_H) \leq 6m t^{15/64}$ whp.
		%\vspace*{-.7\baselineskip}
\end{claim}
\begin{proof}
	Let $H_i$ denote the subgraph of $H$ induced by $S_i$, and let $H_{ij}$ denote the subgraph of $H$ containing only edges with one vertex in $S_i$ and the other in $S_j$. That is, write $\mA_H$ in the following way:
	\[ \mA_H = \begin{bmatrix} \mH_1 & \mH_{12} & \mH_{13} \\ \mH_{21} & \mH_2 & \mH_{23} \\ \mH_{31} & \mH_{32} & \mH_3   \end{bmatrix}, \]
	and use this to bound the maximal eigenvalue of $\mA_H$ as
	\begin{equation}%
	\begin{aligned}
	\lambda_1(\mA_H) 
	&= \lambda_1\left( \begin{bmatrix} \mH_1 & \mH_{12} & \mH_{13} \\ \mH_{21} & \mH_2 & \mH_{23} \\ \mH_{31} & \mH_{32} & \mH_3   \end{bmatrix}\right) \\
	& \leq \lambda_1(\mH_1) + \lambda_1(\mH_2) + \lambda_1(\mH_3) + \lambda_1(\mH_{12}) + \lambda_1(\mH_{23}) + \lambda_1(\mH_{13}). \nonumber
		\end{aligned}
	\end{equation}%
	Note that the maximum eigenvalue of a graph is at most the maximum degree of a graph. By Claim \ref{Claim6} with $f(t) = t_1$ and $\varepsilon = 1/64$, 
	\begin{equation}%
	\begin{aligned}
	&\lambda_1(\mH_1) \leq \Delta_1(\mH_1) = \max_{v \leq t_1} \{ d_{t_1}^m(v)\} \leq t_1^{p/2 + \varepsilon} \leq t^{33/512} \\
	&\lambda_1(\mH_2) \leq \Delta_1(\mH_2) \leq \max_{t_1 \leq v \leq t_2} \{ d_{t_2}^m(v)\} \leq t_2^{p/2 + \varepsilon} / t_1^{p/2} \leq t^{233/1024} \\
	&\lambda_1(\mH_3) \leq \Delta_1(\mH_3) \leq \max_{t_2 \leq v \leq t_3} \{ d_{t_3}^m(v)\} \leq t_3^{p/2 + \varepsilon} / t_2^{p/2} \leq t^{15/64} \nonumber
	\end{aligned}
	\end{equation}%
To bound $\lambda_1(\mH_{ij})$, start with $m = 1$. For $i<j$, this implies that each vertex in $S_j$ has at most one edge in $\mH_{ij}$, i.e. $H_{ij}$ is a star forest. Then we have a bound on $\mH_{ij}$ by Claim \ref{Claim8}. For $m > 1$, let $G'$ be one of our generated graphs with $t$ edges and $m = 1$. Think now of contracting vertices in $G'$ (only the ones added using preferential attachment) into a single vertex. We can write $\mA_G$ in terms of $\mA_G'$: 
		$\mA_G = \mC^T \mA_G' \mC,$
 where $\mC$ is a contraction matrix with $t$ rows and the number of columns equal to the number of vertices in $\mA_G$ (at most $t/m$). The $i$\nth column is equal to $1$ at indices $j$ in which $(i,j)$ are identified. Similarly, we can write $\mH_{ij}$ in terms of $\mH_{ij}'$. 
	
Note that if $\vy = \mC \vx$, then $\vy^T \vy = \vx^T \mC^T \mC \vx$, where $\mC^T \mC$ is a diagonal matrix with $1's$ and $m'$s on the diagonal. So $\vx^T \vx \leq \vy^T \vy \leq m \vx^T \vx.$
	\begin{equation}%
	\label{eqn:HH'}
	\begin{aligned}
	\lambda_1(\mH_{ij}) &= \max_{\vx \neq 0} \frac{ \vx^T \mH_{ij} \vx}{\vx^T \vx} 
	= \max_{\vx \neq 0} \frac{ \vx^T \mC^T \mH_{ij}' \mC \vx}{\vx^T \vx} 
	=  \max_{\vx \neq 0, \vy = \mC \vx} \frac{ \vy^T  \mH_{ij}' \vy}{\vx^T \vx} \\
	&=  \max_{\vx \neq 0, \vy = \mC \vx} \frac{ m \vy^T  \mH_{ij}' \vy}{m \vx^T \vx}
	\leq  \max_{\vx \neq 0, \vy = \mC \vx} \frac{ m \vy^T  \mH_{ij}' \vy}{\vy^T \vy}
	\end{aligned}
	\end{equation}%
Now using Claim \ref{Claim6} with $f(t) = t_1$ and $\varepsilon = 1/64$,
\begin{equation}%
\label{eq:DegH}
\begin{aligned}
&\Delta_1(\mH_{12}') = \max_{\vv \leq t_2} \{ d_{t_2}'(v)\} \leq t_2^{p/2 + \varepsilon} \leq t^{297/1024} \\
&\Delta_1(\mH_{23}') = \max_{t_1\leq v \leq t_3} \{ d_{t_3}'(v)\} \leq t_3^{p/2 + \varepsilon} / t_1^{p/2} \leq t^{29/64} \\
\end{aligned}
\end{equation}%
Finally, all edges in $\mH_{13}'$ are between $S_1$ and $S_3'$, so Claim \ref{Claim7} shows $\Delta_1(\mH_{13}') \leq t^{p - 9/16} \leq t^{7/16}$ whp. Putting together Equations~\eqref{eqn:HH'} and~\eqref{eq:DegH}, we get
$ \lambda_1(\mH_{ij}) \leq m\lambda_1(\mH_{ij}') \leq m \Delta_1(\mH_{ij}')^{1/2} \leq mt^{15/64} $.
And so we get the final bound,
\[\lambda_1(\mA_H) \leq \sum_{i = 1}^3 \lambda_1(\mH_i) + \sum_{i < j} \lambda_1(\mH_{ij}) \leq 6mt^{15/64} \]
This shows that $\lambda_i(\mA_H)$ is $o(\lambda_k(\mA_F))$, which implies $\lambda_i(\mA_G) = \lambda_i(\mA_G) = \lambda_i(\mA_F)(1 \pm o(1))$. \hspace{1.5in} \rlap{$\qquad \Box$}
\end{proof}%
\end{proof}%

\section{TGPA}
\label{sec:TGPA} 

\begin{marginfigure}[50\baselineskip]
	{\includegraphics[width=\linewidth]{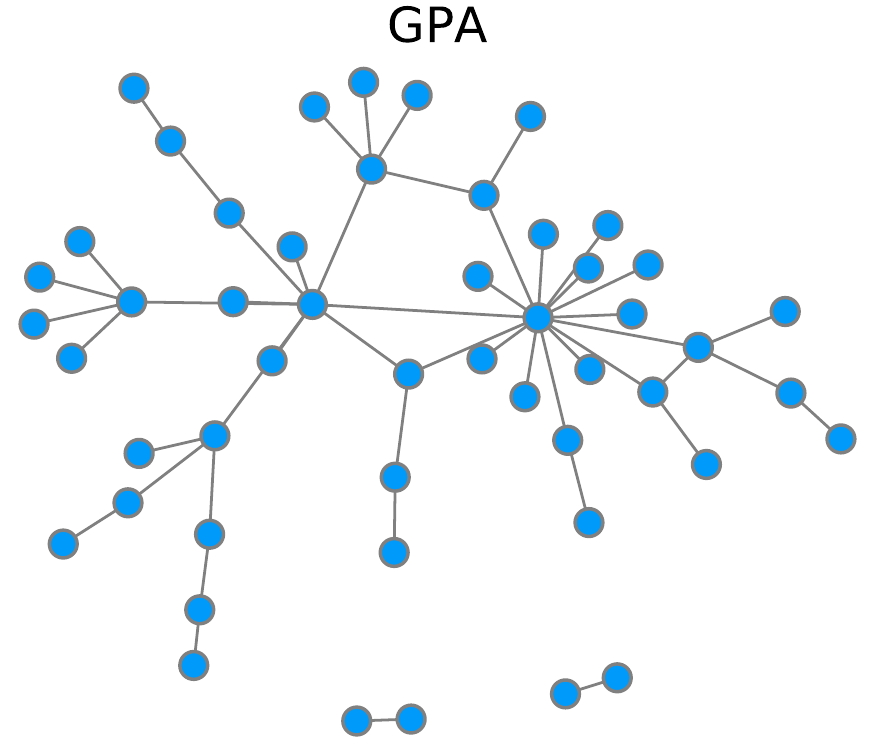}} \\ \vspace{0.2in}
	{\includegraphics[width=\linewidth]{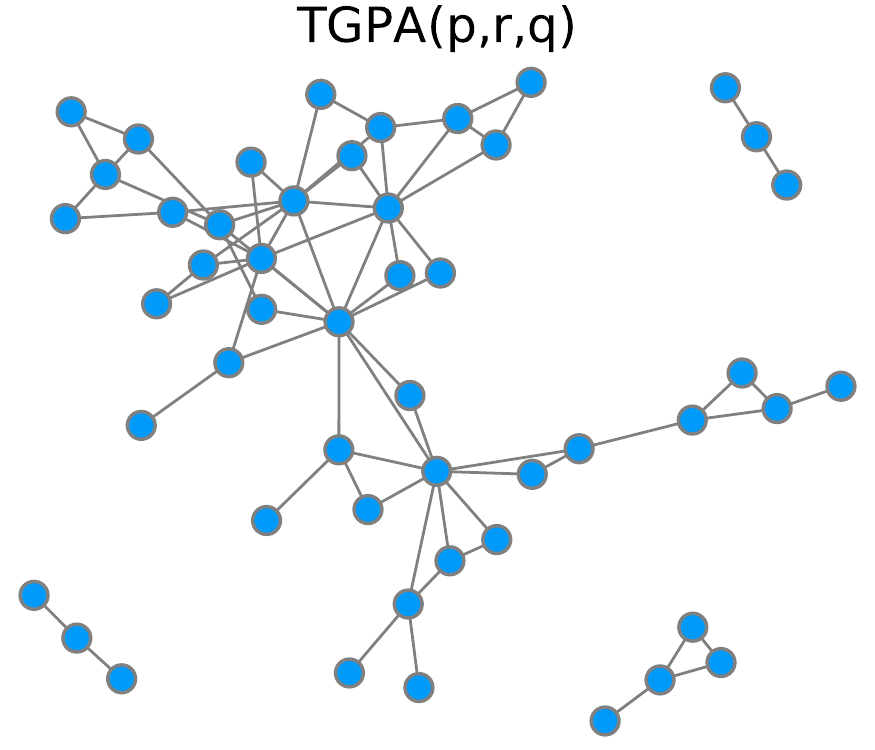}} \\  \vspace{0.2in}
	{\includegraphics[width=\linewidth]{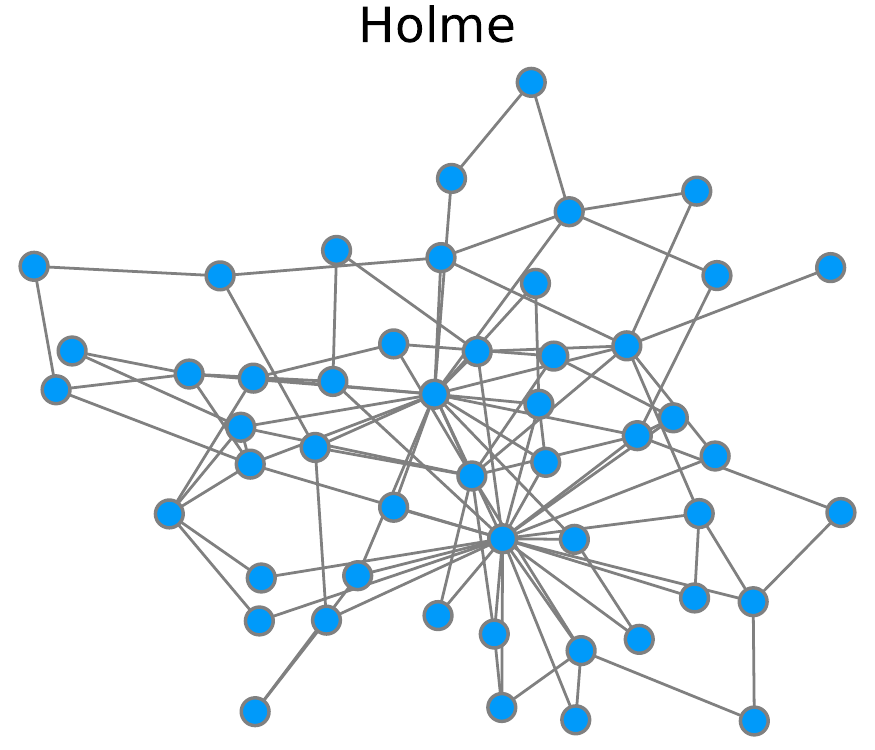}} \\  \vspace{0.2in}
	{\includegraphics[width=\linewidth]{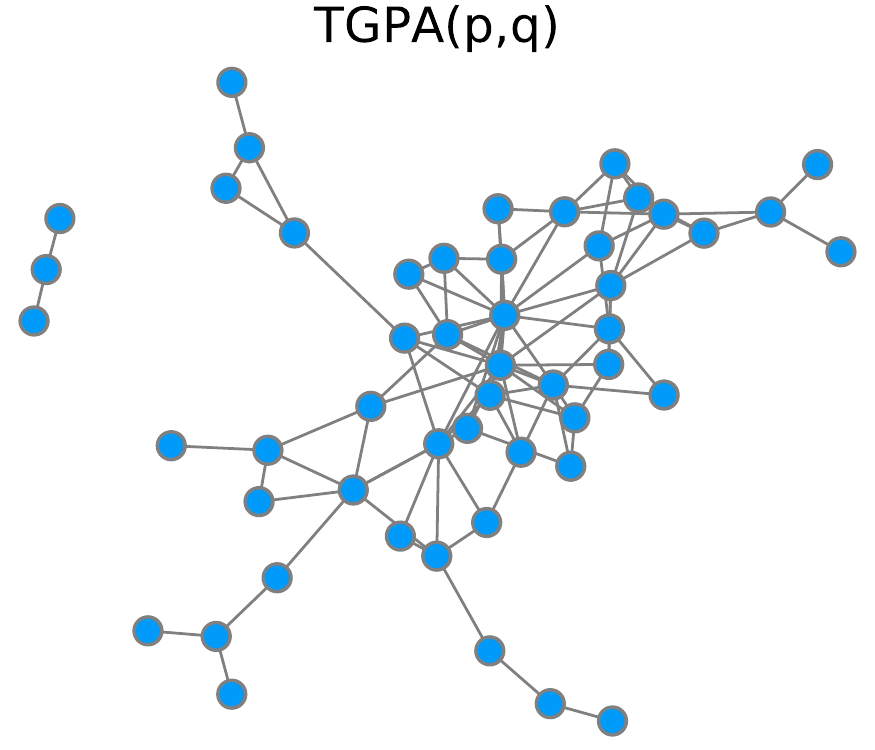}}
	\caption{Examples of 50 node graphs. The top two figures were generated using $p = 0.8, r = 0.1, q = 0.1$. The graphs on the bottom were generated using $m = 2$, and TGPA($p,q$) used $p = 0.85$. See the text for the details on these parameters.}	
	\label{fig:examples}
\end{marginfigure}

In this section we present our model which we call Triangle Generalized Preferential Attachment (TGPA). This model is motivated by the purpose of adding higher order structure into the resulting graph as discussed in Section~\ref{sec:HigherOrder}, and a recent paper~\cite{avin2017improved} which shows a model of preferential attachment with any power-law exponent (Section~\ref{sec:GPA}). We present two different versions of the model. The first, in Section~\ref{sec:v1} follows the PA model as described by~\citet{barabasi1999emergence,flaxman2005high}, and the second in Section~\ref{sec:v2} follows the PA model as described in~\citet{chung2006complex,avin2017improved}. 
Though these models are not the same, they share similar properties. In Sections~\ref{sec:SpectraTGPA} and ~\ref{sec:DegreeTGPA} we'll see each formulation is useful for the analysis of the models. Figure~\ref{fig:examples} shows some example graphs generated by TGPA compared to existing models.

\subsection{TGPA($p,q$)}
\label{sec:v1}
Start with an empty graph. At time $t = 1,2, \ldots$ do one of the following:
\begin{enumerate}
	\item (node event) With probability $p$, add a new vertex $v_t$, and an edge from $v_t$ to some other vertex $u$ where $u$ is chosen with probability
		\begin{equation}%
		\label{eq:PA}
		Pr[u = v_i] = \left\{  \begin{array}{ ll } 
		\frac{d_t(v_i)}{4t-2}, & \text{if} \,\, v_i \neq v_t \\
		\frac{2}{4t-2}, & \text{if}\,\, v_i = v_t
		\end{array} \right.
		\end{equation}%
		Then pick a neighbor of $u$, call it $w$, and also add an edge from $v_t$ to $w$. We pick $w$ with the following probability: 
		\begin{equation}%
		\label{eq:PA2}
		Pr[w = v_i] = \left\{  \begin{array}{ ll } 
		\frac{\text{\# edges between}\,\,  u,w}{d_{t-1}(u)}, & \text{if} \,\, v_i \neq u \\
		\frac{2\cdot \text{\# self-loops of}\, \, u}{d_{t-1}(u)}, & \text{if}\,\, v_i = u
		\end{array} \right.
		\end{equation}%
\item (component event) With probability $q = 1-p$ add a wedge to the graph (3 new nodes with 2 edges)
\item For some constant $m$, every $m$ steps contract the most recently added vertices through the preferential attachment steps (in step 1) to form a super vertex.
\end{enumerate}

Note that vertex $w$ (chosen in step 1) is \emph{also} chosen via preferential attachment. The probability of picking $w$ is the probability of picking $u$ as a neighbor of $w$ times the probability of picking $w$. 
\begin{equation}%
\begin{aligned}%
Pr[w = v_i] &= \frac{ \sum_{u \in N(w)} d_{t-1} (u) }{4t-2} \cdot \frac{\text{num edges between}\,\,  u,w}{d_{t-1}(u)} = \frac{d_t(w)}{4t-2} \nonumber
\end{aligned}%
\end{equation}%

\subsection{TGPA($p_t, r_t, q_t$)}
\label{sec:v2}
Start with a graph with $e_0$ edges. At time $t = 1,2, \ldots$ do one of the following:
\begin{enumerate}
	\item (node event) With probability $p_t$, add a new vertex $v_t$, and an edge from $v_t$ to some other vertex $u$ where $u$ is chosen with probability given in Equation~\eqref{eq:PA}.
%	\begin{smalleq}%
%	\begin{equation}%
%	%\gamma_t(u) = \frac{ d_{t-1}(u)}{\sum_{w \in V_{t-1} d_{t-1}(w)}}
%	Pr[u = v_i] = \left\{  \begin{array}{ ll } 
%	\frac{d_t(v_i)}{4t-2}, & \text{if} \,\, v_i \neq v_t \\
%	\frac{2}{4t-2}, & \text{if}\,\, v_i = v_t
%	\end{array} \right. \nonumber
%	\end{equation}%
%	\end{smalleq}%
	Then pick a neighbor of $u$, call it $w$, as in Equation~\eqref{eq:PA2}. Add edge an edge from $v_t$ to $w$. 
%	\begin{smalleq}%
%	\begin{equation}%
%	Pr[w = v_i] = \left\{  \begin{array}{ ll } 
%	\frac{\text{\# edges between}\,\,  u,w}{d_{t-1}(u)}, & \text{if} \,\, v_i \neq u \\
%	\frac{2\cdot \text{\# of self-loops of}\, \, u}{d_{t-1}(u)}, & \text{if}\,\, v_i = u
%	\end{array} \right. \nonumber
%	\end{equation}%
%	\end{smalleq}%
	\item (wedge event) With probability $r_t$ add a wedge to the graph by picking two nodes using preferential attachment: $v_1, v_2$. Pick the third node uniformly from a neighbor of $v_1$, call it $w$. Add edges $(v_1,v_2)$ and $(v_1,w)$.
	\item (component event) With probability $q_t$ add a wedge to the graph (3 new nodes with 2 edges).
\end{enumerate}

%\subsection{Novelty}
%While this model has many similarities to existing models, laid out in Sections~\ref{sec:GPA}-~\ref{sec:Holme}, we think our model offers new improvements:
%\begin{compactenum}
%	\item Our model combines the benefits from each of these models.
%	\item Our model can be generalized to add other types of components if desired.
%	\item Our model gives new interesting types of graphs, as displayed in figure \ref{fig:examples}.
%	\item We give a lot of analysis on power-law exponents in the degrees and eigenvalues.
%\end{compactenum}

%\input{20-GPAresults}
\section{Analysis of TGPA$(p,q)$}
\label{sec:SpectraTGPA}

In this section we present results on the degrees and spectra of the TGPA$(p,q)$ model described in Section~\ref{sec:v1}. We do not prove these results, however the proofs follow the proof techniques presented in Section~\ref{sec:GPAResults}. The key difference in these proofs is the fact that two edges may be added in each time step. This makes the preferential attachment much more tedious to track through graph generation. In Lemma~\ref{TGPALemma2} for example, we consider disjoint (but not \emph{disconnected}) sets of supernodes; the probability of the supernodes increasing in degree is not independent from one other. 

Fix parameter $p$. Denote $G_t^m(p,q)$ as the Triangle Generalized Preferential Attachment Graph at time $t$ with contractions of size $m$.
\begin{lemma}
	\label{TGPALemma1}
	Let $d_t(s)$ be the degree of vertex $s$ in $G_t^m(p,q)$, for any time $t$ after $s$ has been added to the graph. Let $a^{(\overline{k})} = a(a+2)(a+4)\cdots (a+k-2)$ be a modified rising factorial function. Let $s'$ be the time at which node $s$ arrives in the graph. Then for any positive integer $k$,
		\[  \mathbb{E}[(d_t(s))^{(\overline{k})} ]\leq (4m)^{(\overline{k})} 2^{pk} \left(  \frac{t}{s'}\right)^{pk}    \]
\end{lemma}

\begin{lemma}
	\label{TGPALemma2}
	Let $S = (S_1, S_2, \ldots, S_l)$ be a disjoint collection of supernodes at time $t_0$. Assume that the degree of $S_i$ at time $t_0$ is $d_{t_0} (S_i) = d_i$. Let $t$ be a time later than $t_0$. Let $p_S(\vr;\vd,t_0,t)$ be the probability that each supernode $S_i$ has degree $r_i + d_i$ at time $t$. Let $d = \sum_{i = 1}^l d_i, r = \sum_{i = 1}^{l} r_i$. If $d = o(t^{1/2})$ and $r = o(t^{2/3})$, then
		\[p_S(\vr; \vd, t_0, t) \leq \left( \prod_{i = 1}^l {r_i + d_i-1 \choose d_i -1} \right) \left( \frac{t_0}{t-1} \right)^{pd} \text{exp} \left\{ 3 + 2t_0 - pd + \frac{19pr}{t^p_{}} \right\} \]
\end{lemma}

\begin{theorem}
	\label{TGPATheorem1}
	Let $m$, $k$ be fixed positive integers, and let $f(t)$ be a function with $f(t) \rightarrow \infty$ as $t \rightarrow \infty$. Let $\Delta_1 \geq \Delta_2 \geq \ldots \geq \Delta_k$ denote the degrees of the $k$ highest degree vertices of $G_t^m(p,q)$. Then
		\begin{equation}%
		\frac{t^{p}}{f(t)} \leq \Delta_1 \leq t^{p}f(t) \, \, \, \text{and} \, \, \,
		\frac{t^{p}}{f(t)} \leq \Delta_i \leq \Delta_{i-1} - t^{p}f(t)   \nonumber
		\end{equation}%
	for $i = 1,2,\ldots,k$ whp.
\end{theorem}

%The bounding of the largest degree nodes in Theorem \ref{TGPATheorem1} implies a \emph{power-law} distribution for the largest degrees.  A set of values $x_1, \ldots, x_k$ satisfies a power-law if it is drawn from a probability distribution where $p(x) \propto x^{-\alpha}$ for some $\alpha$. The factor of $t^{p}$ in Theorem \ref{TGPATheorem1} implies a power-law distribution with exponent $\alpha = (1 + p)/p$. This can be seen by using a martingale argument, which has been done a number of times. See for instance~\cite{van2016random}. Notice that depending on the value chosen for $p$, we can obtain a power-law fit with exponents ranging between $2$ and $\infty$. 

%We aim to prove a result for the GPA model which relates maximum eigenvalues and maximal degrees. It is similar to results found in~\cite{MihailPapadimitriou2002,ChungLuVu2003,chung2003spectra}, but follows almost exactly from~\cite{flaxman2005high}.
\begin{theorem}
	\label{TGPATheorem2}
	Let $k$ be a fixed integer, and let $f(t)$ be a function with $f(t) \rightarrow \infty$ as $t \rightarrow \infty$. Let $\lambda_1 \geq \lambda_2 \geq \ldots \geq \lambda_k$ be the $k$ largest eigenvalues of the adjacency matrix of $G_t^m(p,q)$. The for $i = 1, \ldots, k$, we have $\lambda_i = (1 + o(1))\Delta_i^{1/2}$, where $\Delta_i$ is the $i$\nth largest degree.
\end{theorem}
\section{Analysis of TGPA$(p_t,r_t,q_t)$}
\label{sec:DegreeTGPA}

%Consider $TGPA(p_t, r_t, q_t)$, where a node event occurs with probability $p_t$, a wedge event occurs with probability $r_t$ and a component event occurs with probability $q_t$, where $p_t + r_t + q_t = 1$ for each $t$. We will restrict the ways in which the parameters can evolve in Section~\ref{sec:pltgpa}.
Consider $TGPA(p_t, r_t, q_t)$, which was described in Section~\ref{sec:v2}. The parameters $p_t, r_t, q_t$ can change over time, though we will restrict the ways in which the parameters can evolve in Section~\ref{sec:pltgpa}.

\subsection{Recursive relation for $m_{k,t}$}
Recall that $m_{k,t}$ is the number of nodes at time $t$ with degree $k$. We wish to write down a relationship for $m_{k,t+1}$ in terms of $m_{k',t}$ for $k' \leq k$. Recall also that the number of edges at time $t$ is $e_t = e_0 + 2t$, and the total sum of degrees at any time $t$ is $2e_t$. Note that for this reason we need only focus on $m_{k,t}$ for $1 \leq k \leq 2e_t$. 

Let $\mathcal{F}_t$ denote the $\sigma$-algebra generated by the graphs $G_0, G_1, \ldots, G_t$ ($\mathcal{F}_t$ holds the history of events up until time $t$). Fix $k \geq 2$. Since $0 \leq d_{t+1}(v) - d_t(v) \leq 4$ for every node $v$ and time $t$, we have 
\begin{equation}%
\label{eqn:expected1}
\EX[m_{k,t+1} | \mathcal{F} ] = \sum_{ \{ v: k-4 \leq d_t(v) \leq k  \} } \PR[d_{t+1}(v) = k].
\end{equation}%
Recall $\gamma_t(v)$ from Equation~\eqref{eqn:gamma}. Denote $\theta_t(v)$ as 2 times the number of self loops in which $v$ is involved divided by $\sum_{w \in V_{t-1}} d_{t-1}(w)$. (i.e. the proportion of edges which are self loops on $v$). If $d_{t+1}(v) = 4$, then there are at most 5 possible values for $d_t(v)$ when $k \geq 4$:
\begin{enumerate}[(i)]
	\item $d_t(v) = k$. In this case there must have either been a node event not involving $v$ (this occurs with probability $p_{t+1}(1 - 2\gamma_{t+1}(v) + \theta_{t+1}(v))$), or a wedge event not involving $v$ (with probability $r_{t+1}(1 - \gamma_{t+1}(v))(1 - 2\gamma_{t+1}(v) + \theta_{t+1}(v))$), or a component event (with probability $q_{r+1}$).
	\item $d_t(v) = k-1$. In this case there must have either been a node event where $v$ is involved as the first node (probability $p_{t+1} \cdot \gamma_{t+1}(v) \cdot (1-\theta_{t+1}(v))$), or where $v$ is involved as the second node (probability $p_{t+1} (\gamma_{t+1}(v) - \theta_{t+1}(v))$), or a wedge event in which $v$ is involved as the first node (with probability $r_{t+1} ( \gamma_{t+1}(v) - \gamma_{t+1}(v)^2 - \theta_{t+1}(v) + \gamma_{t+1}(v) \cdot \theta_{t+1}(v))$) or as the third node (probability $r_{t+1}(1 - \gamma_{t+1}(v))(\gamma_{t+1}(v) - \theta_{t+1}(v))$).
	\item $d_t(v) = k-2$. In this case there must have either been a node event in which $v$ is picked as both nodes involved (with probability $p_{t+1} \cdot \theta_{t+1}(v)$) or there must have been a wedge event in which $v$ is involved as the second node (with probability $r_{t+1} \cdot \theta_{t+1}(v)(1 - \gamma_{t+1}(v))$) or as the first and third nodes (with probability $r_{t+1} \cdot \gamma_{t+1}(v)(1 - \gamma_{t+1}(v) + \theta_{t+1}(v))$).	
	\item $d_t(v) = k-3$. In this case there must have been a wedge event where $v$ was involved as the first and second nodes or there was a wedge event where $v$ was involved as the second and third nodes (these events occur in combination with probability $2r_{t+1}\gamma_{t+1}(v)(\gamma_{t+1}(v) - \theta_{t+1}(v)$ ) ). 
	\item $d_t(v) = k-4$. In this case there must have been a wedge event where $v$ is picked for all three wedges, which happens with probability $r_{t+1} \cdot \gamma_{t+1}(v) \cdot \theta_{t+1}(v)$
\end{enumerate}
Let $\alpha_{k,t} = k/(2e_t)$. Then for every $v$ such that $d_t(v) = i$, $\gamma_{t+1}(v) = \alpha_{i,t}$. Define
\begin{equation}%
\begin{aligned}
A_{k,t} &= p_{t+1 ,k} (1 - 2\alpha_{k,t} + \theta_{t+1}(v))+ r_{t+1} (1 - \alpha_{k,t})(1 - 2 \alpha_{k,t} + \theta_{t+1}(v))  + q_{t+1},\\
B_{k,t} &= 2 p_{t+1} (\alpha_{k,t} - \theta_{t+1}(v)) + 2r_{t+1}(1 - \alpha_{k,t})(\alpha_{k,t} - \theta_{t+1}(v)),  \\
C_{k,t} &= p_{t+1}\theta_{t+1}(v) + r_{t+1} (\alpha_{k,t} - \alpha_{k,t}^2 + \theta_{t+1}(v)),\\
D_{k,t} &= 2 r_{t+1} \alpha_{k,t}( \alpha_{k,t} - \theta_{t+1}(v)), \, \, \text{and} \, \, E_{k,t} = r_{t+1} \alpha_{k,t} \theta_{t+1}(v).\nonumber
\end{aligned}
\end{equation}%
Then $A_{k,t} + B_{k,t} + C_{k,t} + D_{k,t} + E_{k,t} = 1$
and $A_{k,t}, B_{k,t}, C_{k,t}, D_{k,t}, E_{k,t}$ $\geq 0$ for every $0 \leq k \leq 2e_t$. Also, by Equation~\eqref{eqn:expected1}, for every $k \geq 4$
\begin{equation}%
\label{eqn:expected2}
\begin{aligned}
\EX[m_{k,t+1} | \mathcal{F}] &= m_{k,t} A_{k,t} + m_{k-1, t} B_{k-1,t} + m_{k-2,t} C_{k-2,t} \\& \hspace{0.3cm} +  m_{k-3,t} D_{k-3,t} + m_{k-4,t} E_{k-4,t}.
\end{aligned}
\end{equation}%
And for remaining values of $k$ we have
\begin{equation}%
\begin{aligned}
\label{eqn:expected3}
\EX[m_{3,t+1} | \mathcal{F}] &= m_{3,t} A_{3,t} + m_{2,t} B_{2,t} + m_{1,t}C_{1,t} \\
\EX[m_{2,t+1} | \mathcal{F}] &= m_{2,t} A_{2,t} + m_{1,t} B_{1,t} + p_{t+1} + q_{t+1} \\
\EX[m_{1,t+1} | \mathcal{F}] &= m_{1,t} A_{1,t} + 2q_{t+1}.
\end{aligned}
\end{equation}%
Define 
\begin{equation}%
X_{k,t} = 
 \begin{cases} 
m_{k-1, t} B_{k-1,t} + m_{k-2,t} C_{k-2,t}+  m_{k-3,t} D_{k-3,t} + m_{k-4,t} E_{k-4,t}   & k \geq 4 \\
m_{2,t} B_{2,t} + m_{1,t}C_{1,t} & k = 3 \\
 m_{1,t} B_{1,t} + p_{t+1} + q_{t+1} & k = 2 \\
2q_{t+1} & k = 1\\
\end{cases}
\end{equation}%
Then Equations~\eqref{eqn:expected2} and~\eqref{eqn:expected3} can be re-written as
\begin{equation}%
\begin{aligned}
\label{eqn:mk}
%&\EX[m_{k,t+1} | \mathcal{F}] = m_{k,t} A_{k,t} + X_{k,t} \\
&\EX[m_{k,t+1}] = \EX[m_{k,t}] \cdot A_{k,t} + \EX[X_{k,t}]
\end{aligned}
\end{equation}%

\subsection{Degree Power-law in TGPA}
\label{sec:pltgpa}
The following lemma is presented in~\citet{avin2017improved} and is a quick generalization of a result in~\citet{chung2006complex}.
\begin{lemma}[\cite{avin2017improved}]
	\label{lemma:Avin}
Suppose that a sequence satisfies the recurrence relation
%\begin{smalleq}% 
%\begin{equation}%
%	a_{t+1} = \left(1- \frac{b_t}{t+t_1} \right) a_t + c_t \nonumber
%\end{equation}%
%\end{smalleq}%
$a_{t+1} = (1- b_t/(t+t_1) ) a_t + c_t$
	for $t \geq t_0$. Furthermore, let $\{s_t \}$ be a sequence of real numbers with $\displaystyle \lim_{t\rightarrow \infty} s_t/s_{t+1} = 1$, $d_t = t(1-s_t/s_{t+1})$, $\displaystyle \lim_{t\rightarrow \infty} b_t = b$, $\displaystyle \lim_{t\rightarrow \infty} c_t \cdot t/s_t = c$, $\displaystyle \lim_{t\rightarrow \infty} d_t = d$, and $b+d > 1$. Then $\displaystyle \lim_{t\rightarrow \infty} a_t/s_t$ exists and $\displaystyle \lim_{t\rightarrow \infty} a_t / s_t = c/(b+d)$. 
\end{lemma}

The following theorem and corollary prove that $TGPA(p_t,r_t,q_t)$ has a power-law in the degree distribution, which we can analyze.
\begin{theorem}
	\label{thm:M}
	Consider TGPA$(p_t,r_t,q_t)$. Let $y_t = p_t + 3q_t$. Assume that $\displaystyle \lim_{t\rightarrow \infty} y_t = y < 3$, $\sum_{ t=1}^{\infty} y_t = \infty$, and $ \displaystyle \lim_{t\rightarrow \infty} t \cdot y_{t+1}/$ $\sum_{j =1}^t y_j = \Gamma > 0$. Then letting $ \beta = 1 + 2 \Gamma/(3-y)$, the limit $\displaystyle M_k = \lim_{t\rightarrow \infty} \EX[m_{k,t}] / \EX[n_t] $ exists for every $k \geq 1$ and 
	\[ M_k = \frac{ \Gamma}{ \Gamma + 3/2 - y/2} \prod_{j =1}^{k-1} \frac{j}{j+ \beta}.  \]
	\vspace*{-.5\baselineskip}
\end{theorem}%
\begin{proof}
	This proof will be an induction on $k$. For $k=1$ we use Lemma~\ref{lemma:Avin} setting $(t_1, s_t, a_t, b_t, c_t) = (e_0, \EX[n_t], \EX[m_{1,t}], e_t(1-A_{1,t}), y_{t+1} ).$ Using Equation~\eqref{eqn:mk}, this gives the limits $b = 3/2 - y/2$, and $c = d = \Gamma$, which concludes the base case. Now assume the Theorem holds for $k-1$, we now prove it for $k$. Again use Lemma~\ref{lemma:Avin}, this time with $(t_1, s_t, a_t, b_t, c_t) = (e_0, \EX[n_t], \EX[m_{k,t}], B_{k-1,t} \EX[m_{k-1,t}] + C_{k-1,t} \EX[m_{k-2,t}] + D_{k-3,t} \EX[m_{k-3,t}] + E_{k-4,t} \EX[m_{k-4,t}] )$. Then we get $d = \Gamma$, $b = k\cdot(3/2 - y/2)$, and using the inductive hypothesis, 
	\[ c = \lim_{t\rightarrow \infty} \frac{c_t \cdot t}{s_t} = (k-1) \left(\frac{3}{2} - \frac{y}{2} \right) M_{k-1} . \]
		Therefore $M_{k}$ exists and 
		\[  M_k = \frac{ (k-1)(3/2 - y/2) M_{k-1} }{k (3/2 - y/2) + \Gamma} = \frac{k-1}{k-1 + \beta} M_{k-1} . \] %\hspace{0.5in}
\end{proof}

The proof of the following corollary follows exactly from~\citet{avin2017improved}.
\begin{corollary}
	\label{Avin2}
	Under the assumptions in Theorem \ref{thm:M}, $M_k$ is proportional to $k^{-\beta}$.
\end{corollary}

Finally, we can state which power-law exponents are obtainable.
\begin{lemma}
	For any $x \in (1,\infty)$, there exists a choice of $p_t, r_t, q_t$ such that in $TGPA(p_t,r_t,q_t)$ the resulting network follows a power-law in the degree distribution with exponent $\beta = x$.
	%\vspace*{-0.5\baselineskip}
\end{lemma}
\begin{proof} 
We can use three separate cases:
\begin{enumerate}[(i)]
	\item For $x \in (5/3,\infty)$, setting $y_t = 3 - 2/(x-1)$ gives exponent
		 $\beta = 1 + 2/( 3 - (3 - 2/(x-1)) ) = x$.
	\item For $x \in (1,5/3)$, set $y_t = t^{3/2(x - 5/3)}$. Then
	\begin{equation}%
	\begin{aligned}
	\Gamma &= \lim_{t\rightarrow \infty} \frac{ y_{t+1} \cdot t }{ \sum_{j = 1}^t y_j} 
	= \lim_{t\rightarrow \infty} \frac{ (t+1)^{3/2(x - 5/3)} \cdot t }{ \sum_{j = 1}^t (j^{3/2(x-5/3)}) } = \lim_{t\rightarrow \infty} \frac{\cdot t^{3/2x - 3/2} }{\int_{j = 0}^t j^{3/2(x - 5/3)} dj } \\
	& = \lim_{t\rightarrow \infty} \frac{ (3/2 x - 3/2) t^{3/2x - 3/2} }{ j^{3/2x - 3/2} |_{j = 0}^t } = 3/2x - 3/2 \nonumber
	 \end{aligned}
	 \end{equation}%
	 Then $\beta = 1 + (2 \Gamma)/(3-y) = 1 + 2(3/2x - 3/2)/(3-0) = x$.
	 \item For $x = 5/3$, set $y_t = 1/\ln(t+2)$ for every $t$. Then we have
	 \begin{equation}%
	 \begin{aligned}
	 \Gamma = \lim_{t\rightarrow \infty} \frac{ y_{t+1} \cdot t }{ \sum_{j = 1}^t y_j} = \lim_{t\rightarrow \infty} \frac{ t / \ln(t + 3) }{\sum_{j =1}^t 1 / \ln(j+2) }
	 = \lim_{t\rightarrow \infty} \frac{ t / \ln(t+3) }{t / \ln{t} } = 1 \nonumber
	 \end{aligned}
	 \end{equation}%
Then TGPA$(p_t,r_t,q_t)$ follows a power law degree distribution with exponent $\beta = 1 + 2 \Gamma / (3-y) = 1 + 2/(3-0) = 5/3$. 
\end{enumerate}%
\end{proof}%

For a final analysis, we show that the component portion is necessary to obtain the full power-law exponent range $(1,\infty)$. Lemma~\ref{lemma:Avin3} comes directly from~\citet{avin2017improved}.

\begin{lemma}[~\cite{avin2017improved}]
	\label{lemma:Avin3}
	Assume $\displaystyle \lim_{t\rightarrow \infty}y_t = y$ and $\displaystyle \lim_{t\rightarrow \infty} y_{t+1} \cdot t/$ $\sum_{j =1}^t j_j = \Gamma$. Then for $y >0$ we have $\Gamma = 1$, and for $y = 0$ we have $\Gamma \leq 1$. 
\end{lemma}

\begin{corollary}
\label{corr:tpga-power-law}
	Consider TGPA($p_t,r_t,q_t$). Assume that $\displaystyle \lim_{t\rightarrow \infty} q_t = 0$, $\displaystyle \lim_{t\rightarrow \infty} y_t = y$, and $\displaystyle y_{t+1} t/$ $\sum_{j = 1}^t y_j = \Gamma > 0$. Then the resulting graph follows a power law degree distribution with exponent $\beta \in (1,3]$.
\end{corollary}

\begin{proof}
	By Corollary~\ref{Avin2}, TGPA$(p_t,r_t,q_t)$ follows a power-law in the degree distribution with exponent $\beta = 1 + 2\Gamma/(3-y) > 1$. By Lemma~\ref{lemma:Avin3}, for $0 < y \leq 1$, we have $\beta = 1 + 2/(3-y) \in (5/3,3]$ and for $y = 0$ we have $\beta = 1 + 2\Gamma/3 \leq 5/3$. 
\end{proof}
\section{significant Clustering coefficients}
\label{sec:clust}

%\begin{figure}
%{\includegraphics[width=0.45\linewidth]{gcc_TGPAvsGPA.pdf}} \hspace{0.17in}
%{\includegraphics[width=0.45\linewidth]{gcc_TGPAvsHOLME.pdf}}
%\caption{A plot of global clustering coefficients in generated graphs with 50,000 nodes. On the left is the $TGPA(p_t,r_t,q_t)$ model vs the $GPA(p_t,r_t,q_t)$ model with $p$ fixed as 0.3, and $q$ varying. On the right is the $TGPA(p)$ model versus the Holme model, with $m = 2$ and $q = 1-p$ varying.}	
%\label{fig:gcc}
%\end{figure}

\begin{table}
	\centering
	\caption{Clustering coefficients for 3 real-world networks, and generated models. TGPA is able to generate data with much larger clustering coefficients, compared to GPA.}
	%\begin{tabularx}{\linewidth}{>{\bfseries}lZZZZZZl}
	\begin{tabularx}{\linewidth}{lZZZZZ}
	\toprule
	\text{Network name} & edges & global clust & local clust & HO global & HO local \\
	\midrule
	\bfseries{\emph{Auburn (18k vertices)}} & 974k & 0.137 & 0.223 & 0.107 & 0.172 \\
	TGPA(18k,0.987,10,150): & 640k & 0.25 & 0.22 & 0.118 & 0.03 \\
	GPA(18k,0.001,0.999,2): & 906k & 0.021 & 0.030 & 0.005 & 0.014 \\
	%Holme(n,): & & & & & \\
	\addlinespace
	\bfseries{\emph{Berkeley (13k vertices)}} & 852k & 0.114 & 0.207 & 0.0876 & 0.156\\
	TGPA(13k,0.99, 10, 58) & 502k & 0.104 & 0.185 & 0.034 & 0.025\\
	GPA(13k,0.001,0.999,2) & 502k & 0.024 & 0.034 & 0.005 & 0.015\\
	\addlinespace
	\bfseries{\emph{Princeton (7k vertices)}} & 293k & 0.237 & 0.164 & 0.091 & 0.146 \\
	TGPA(7k,0.987,10,100): & 207k & 0.298 & 0.251 & 0.148 & 0.053 \\
	GPA(7k,0.001,0.999,2): & 255k & 0.038& 0.054& 0.009 & 0.025 \\
	%Holme(n,10,100): & & & & & \\
	\bottomrule
		%\vspace*{-\baselineskip}
\end{tabularx}
\label{table:Clust}
\end{table}

%Here we show the TGPA model has higher-order clustering.

We analyzed 3 networks from the Facebook 100 dataset~\cite{Traud-2012-facebook}, each of which is a set of users at a particular university. We computed the global clustering coefficient: $6 |K_3|/|W|$ where $|K_3|$ is the number of triangles and $|W|$ is the number of wedges, and average local clustering coefficient: the average of $2|K_3(u)|/|W(u)|$ for all nodes $u$, where $K_3(u)$ denotes triangles for which $u$ is a member. We also considered \emph{higher-order} clustering coefficients, defined in~\citet{yin2017higher} to be the fraction of appropriate motifs which are closed into 4-cliques. 

To fit the TGPA$(p,q)$ model (Section~\ref{sec:v1}) to the real world networks, we noted that the average degree of our model, the total degrees divided by the number of nodes, is approximately $(2m(1-p) + 2m)/(m(1-p) +1)$. Choosing the average degree gives a relationship between parameters $m$ and $p$. We tested various sets of parameters to obtain the best possible fit. We started both TGPA and GPA with a $k$-node clique. Table~\ref{table:Clust} lists the parameters we chose for the TGPA model as TGPA($n, p, k, m$), which produces an $n$ node graph starting from a $k$ node clique. For comparison we also fit the GPA model (Section~\ref{sec:GPA}). The parameters in Table~\ref{table:Clust} are GPA($n,p,r,k$). Notice that TGPA maintains much more significant clustering coefficients across all measures.

\section{The Eigenvalue Power-law is robust}
\label{sec:per}
\begin{marginfigure}
	\begin{center}
		\includegraphics[width=\linewidth,clip, trim = 0 0 0 1cm]{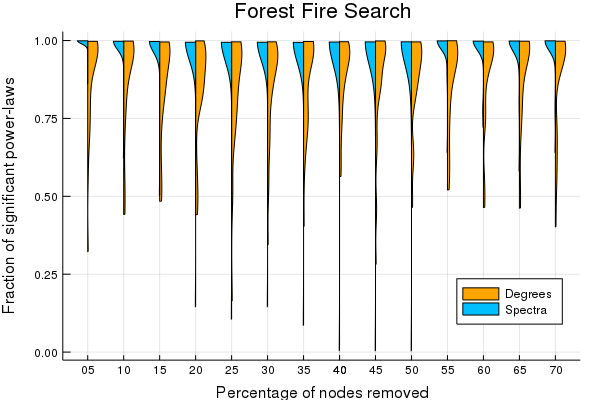}
		\caption{Forest Fire Sampling graphs generated using the preferential attachment model. }
		\label{fig:ffpa}
	\end{center}
\end{marginfigure}
As discussed at length in this paper, preferential attachment has long been used to describe the reason why we find power-law distributions in the degrees of real world networks. There are many other empirical and theoretical studies on the presence of power-laws in spectra~\cite{ChungLuVu2003, goh2001spectra, MihailPapadimitriou2002,Eikmeier-2017-power-laws}. Given that many real-world networks should have power-laws in both the eigenvalues and the degrees, this suggests that one should be easier and more reliable to detect than the other. Our recent paper~\cite{Eikmeier-2017-power-laws} gives evidence that power-laws are more likely to be present in the spectra than in degree distributions. An explanation for this observation may come from the way in which we obtain data, rather than a true feature of the data itself. Consider for a moment that the ``real data'' that is used in so many studies is not the full set of data. Instead, due to sampling or missing data the ``real data'' is actually some perturbation of the true set. If the underlying graph has both a power-law in the degrees and eigenvalues, then it is possible this observation just reflects the robustness of the eigenvalue power-law to the type of network sampling that occurred. There are many methods of sampling graphs, and studying properties of sampled graphs is a well-studied field~\cite{leskovec2006sampling,stumpf2005sampling,stumpf2005subnets,lovasz1993random,lee2006statistical,ebrahimi2017complex,schoenebeck2013potential}.

\begin{fullwidthfigure}
	\begin{center}
		\includegraphics[width = 0.32\linewidth]{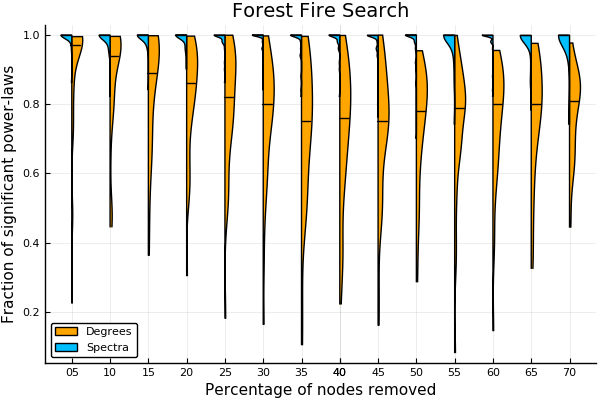} \hfill
		\includegraphics[width = 0.32\linewidth]{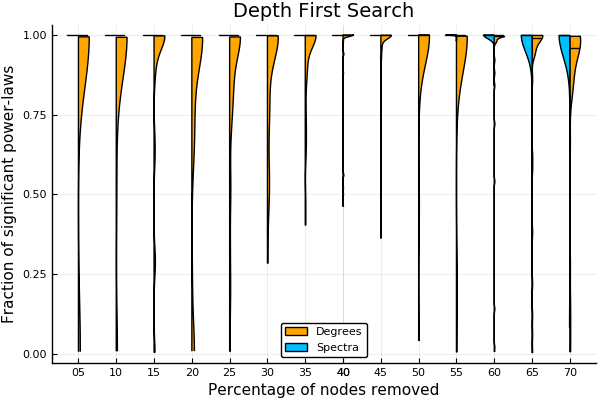} \hfill
		\includegraphics[width = 0.32\linewidth]{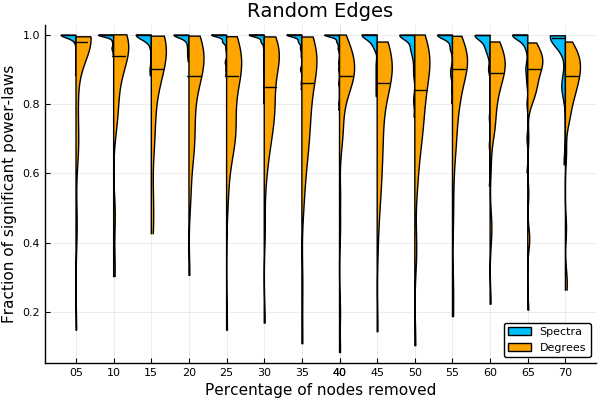}
		\caption{35 TGPA graphs with power-law exponents between 2-5, sub-sampled in various ways. On the top, the graphs were sampled using a forest fire search on a random seed node; in the middle a depth first search;  and on the bottom, by sampling random edges. Note that when there appears to be \emph{no} violin plot (e.g. most spectra in DFS) that means 100$\%$ of the sampled graphs had significant power-laws. The horizontal lines give the median.}
		\label{fig:per}
	\end{center}
\end{fullwidthfigure}

Because the TGPA model produces graphs with reliable power-law exponents in \emph{both} the degrees and spectra, as well as clustering,  (Sections~\ref{sec:DegreeTGPA},~\ref{sec:SpectraTGPA}) this makes it a good model to study this potential effect. We generated 35 TGPA graphs of size 5000 with theoretical degree power-law exponents between $2-5$. For each graph, we detect that it has a statistically significant power-laws in both the degrees and spectra. The distributions were tested for power-laws using the method of~\citet{ClausetShaliziNewman2009}. We then perturbed each of the networks in three ways: In the first method we sampled random sets of edges of the graphs; in the second method we did a depth first search, starting at a random seed node; and third we did a forest fire sampling procedure from a random seed node (at each time step a fire ``spreads'' to each neighbor with some probability based on a burn rate). In each case, we ran the perturbation until a certain percentage of the nodes were obtained. And in each case we did the perturbation 50 times. The results of this experiment are shown in violin plots in Figure~\ref{fig:per}. Notice that the degree plots have a much larger spread in most cases, and the spectra almost always retains its power-law. 

TGPA isn't the only model with power-laws in both the degrees and eigenvalues (as we've proved about the GPA model in this manuscript). So we ask if we see these same sampling effects on \emph{other} classes of models. When trying the same experiment on PA models, we don't see as much variation between the degrees and spectra. See Figure~\ref{fig:ffpa} for an example of the forest fire sampling procedure. The other sampling procedures give similar results. We believe that the local structure of TGPA is necessary to see the effects of sampling. Note that we used the same size graphs, same number and number of samples as in the TGPA experiment.

\section{Conclusions and Discussion}
\label{sec:conclusion}

In this paper we presented the triangle generalized preferential attachment model, a graph model which incorporates direct triangle formulation into the preferential attachment model. Furthermore, we provided extensive analysis of this model, showing that the degree and spectral distributions fit power-law distributions. We also provided extended analysis of the generalized preferential attachment model found in~\citet{avin2017improved}.

We further showed that triangle generalized preferential attachment has improved clustering coefficients over traditional preferential attachment models. Of course there are other models which exhibit higher order clustering that lack theoretical proofs of power-laws in both degrees and eigenvalues~\cite{Eikmeier-preprint-HyperKron,holme2002growing,lattanzi2009affiliation}. 

Our new model provides a useful platform for studying real-world network data. We found that it provided far more clustering in network data compared with the standard preferential attachment model. We further showed that if a graph has a significant power-law in both the spectra and degrees, under various sampling procedures (forest fire sampling, depth first search, and random edges), the spectral power-law remains much more frequently in sampled data. We provide this experiment as evidence for one possibly reason why we may see power-law distributions in the spectra of real networks more often than in the degrees. In the future, we plan to study further generalizations of higher-order preferential attachment graphs. 
\begin{fullwidth}
\bibliographystyle{dgleich-bib}
\bibliography{All_Docs,99-refs}
\end{fullwidth}

%\clearpage
%\input{90-Appendix}
\end{document}